\documentclass[runningheads]{llncs}
\usepackage{graphicx}
\usepackage{amsfonts}
\usepackage{xcolor}

\usepackage{amsmath}
\usepackage{amssymb}
\usepackage{amsthm}
\usepackage[shortlabels]{enumitem}

\usepackage[boxruled,vlined,nofillcomment]{algorithm2e}

\allowdisplaybreaks

\newcommand{\cD}{\mathcal{D}}
\newcommand{\N}{\mathbb{N}}
\newcommand{\F}{\mathbb{F}}
\newcommand{\bb}{\bx}
\newcommand{\bA}{\mathbf{A}}

\newcommand{\bt}{\mathbf{t}}
\newcommand{\bzero}{\mathbf{0}}
\newcommand{\bx}{\mathbf{x}}
\newcommand{\by}{\mathbf{y}}
\newcommand{\cP}{\mathcal{P}}
\newcommand{\cA}{\mathcal{A}}
\newcommand{\cX}{\mathcal{X}}
\newcommand{\cY}{\mathcal{Y}}
\newcommand{\cZ}{\mathcal{Z}}

\newcommand{\uniform}{\mathrm{Uniform}}
\newcommand{\yy}{\mathbf{y}}
\newcommand{\davg}{d_\textrm{avg}}
\newcommand{\poly}{\text{poly}}

\newcommand{\inputring}{\mathbb{F}_q}%{\mathbb{Z}_q}
\newcommand{\sumset}{\mathcal{B}}
\newcommand{\enc}{\mathsf{Enc}}
\newcommand{\enclocal}[2]{\mathcal{E}_{#1}^{#2}}
\newcommand{\encdist}[2]{\mathcal{E}_{#1}^{#2}}
\newcommand{\encshufdist}[2]{\mathcal{S}_{#1}^{#2}}
\newcommand{\sd}{\mathrm{SD}}
\newcommand{\invset}{\mathrm{Inv}}

\renewcommand{\epsilon}{\varepsilon}
\newcommand{\analyzer}{\mathcal{A}}

\newtheorem{observation}{Observation}
\newtheorem{fact}{Fact}
\DeclareMathOperator{\defc}{defc}
\DeclareMathOperator{\rank}{rank}

\DeclareMathOperator{\supp}{supp}
\DeclareMathOperator*{\E}{\mathbb{E}}
\DeclareMathOperator*{\Var}{Var}
\DeclareMathOperator*{\argmax}{argmax}

\begin{document}

\title{Private Aggregation\\ from Fewer Anonymous Messages}

%\titlerunning{Private Aggregation\\ from Fewer Anonymous Messages}
\author{Badih Ghazi\inst{1} \and
Pasin Manurangsi\inst{1} \and
Rasmus Pagh\inst{1,2}\and
Ameya Velingker\inst{1}}
\authorrunning{B. Ghazi et al.}
\institute{Google Research, Mountain View CA 94043, USA \and
IT University of Copenhagen, Denmark \\ \email{\{badihghazi,pasin,pagh,ameyav\}@google.com}}

\maketitle              % typeset the header of the contribution
\begin{abstract}

%\badih{Abstract instructions from the template: The abstract should briefly summarize the contents of the paper in 150--250 words.}

Consider the setup where $n$ parties are each given an element~$x_i$ in the finite field $\F_q$ and the goal is to compute the sum $\sum_i x_i$ in a secure fashion and with as little communication as possible. 
We study this problem in the \emph{anonymized model} of Ishai et al.~(FOCS 2006) where each party may broadcast anonymous messages on an insecure channel.

We present a new analysis of the one-round ``split and mix'' protocol of Ishai et al.
In order to achieve the same security parameter, our analysis reduces the required number of messages by a $\Theta(\log n)$ multiplicative factor.

We also prove lower bounds showing that the dependence of the number of messages on the domain size, the number of parties, and the security parameter is essentially tight.

Using a reduction of Balle et al. (2019), our improved analysis of the protocol of Ishai et al. yields, in the same model, an $\left(\varepsilon, \delta\right)$-differentially private protocol for aggregation that, for any constant $\varepsilon > 0$ and any $\delta = \frac{1}{\poly(n)}$, incurs only a constant error and requires only a \emph{constant number of messages} per party. Previously, such a protocol was known only for $\Omega(\log n)$ messages per party.

\end{abstract}

% \newpage

% \section{Notation (TO BE REMOVED BEFORE SUBMISSION)}

% \noindent
% \begin{tabular}{c|c|l}
%      Symbol & Macro & Meaning \\
%      \hline
%      $(\varepsilon,\delta)$ & & DP parameters\\
%      $n$ & & Number of parties\\
%      $??$ & & Bit length of integer input\\
%      $m$ & & Number of messages per party\\
%      $\sigma$ & & Logarithm of statistical error\\
%      $q$ & & Field size\\
%      $\inputring$ & \texttt{\textbackslash inputring} & Ring containing party inputs\\
%      $\sumset_s$ & \texttt{\textbackslash sumset} & Set of input vectors whose coordinates sum to $s$ \\
%      $\enc$ & \texttt{\textbackslash enc} & Local encoder \\
%      $\enclocal{x}{\enc}$ & \texttt{\textbackslash enclocal\{x\}\{\textbackslash enc\}} & Distribution of messages output by the encoder on input $x$ \\
%      $\encdist{\bx}{\enc}$ & \texttt{\textbackslash encdist\{\textbackslash bx\}\{\textbackslash enc\}} & Distribution of pre-shuffled output of $\bf{x}$ \\
%      $\encshufdist{\bx}{\enc}$ & \texttt{\textbackslash encshufdist\{\textbackslash bx\}\{\textbackslash enc\}} & Distribution of post-shuffled output of $\bf{x}$ \\
%      $\sd$ & \texttt{\textbackslash sd} & Statistical distance (same as total variation distance) \\
%      $\cloak$ & \texttt{\textbackslash cloak} & Split and mix protocol \\
%      \hline
% \end{tabular}

% \newpage

\section{Introduction}

%\subsection*{Secure Aggregation via Anonymity}
We study one-round  multi-party protocols for the problem of secure aggregation:
Each of $n$ parties holds an element of the field $\F_q$ and we wish to compute the sum of these numbers, while satisfying the security property that for every two inputs with the same sum, their transcripts are ``indistinguishable.'' The protocols we consider work in the \emph{anonymized model}, where parties are able to send anonymous messages through an insecure channel and indistinguishability is in terms of the \emph{statistical distance} between the two transcripts (i.e., this is information-theoretic security rather than computational security). This model was introduced by Ishai et al.~\cite{IKOS06} in their work on cryptography from anonymity\footnote{Ishai et al. in fact considered a more general model in which the adversary is allowed to corrupt some of the parties; please refer to the discussion at the end of Section~\ref{subsec:our-results} for more details.}.
We refer to~\cite{IKOS06,cheu19} for a discussion of cryptographic realizations of an anonymous channel.

The secure aggregation problem in the anonymized model was studied already by Ishai et al.~\cite{IKOS06}, who gave a very elegant one-round ``split and mix'' protocol.
Under their protocol, each party $i$ holds a private input $x_i$ and sends $m$ anonymized messages consisting of random elements of $\inputring$ that are conditioned on summing to $x_i$. Upon receiving these $mn$ anonymized messages from $n$ parties, the server adds them up and outputs the result.
Pseudocode of this protocol is shown as Algorithm~\ref{alg:ishai_et_al}.
%Note that for any value of the number $m$ of messages, this protocol outputs the correct sum.
Ishai et al.~\cite{IKOS06} show that as long as $m$ exceeds a threshold of $\Theta\left(\log n + \sigma + \log q\right)$, this protocol is \emph{$\sigma$-secure} in the sense that the statistical distance between transcripts resulting from inputs with the same sum is at most~$2^{-\sigma}$.

\paragraph{Differentially Private Aggregation in the Shuffled Model.}
An exciting recent development in differential privacy is the \emph{shuffled model}, which is closely related to the aforementioned anononymized model. The shuffled model provides a middle ground between two widely-studied models of differential privacy. 
In the \emph{central model}, the data structure released by the analyst is required to be differentially private, whereas the \emph{local model} enforces the more stringent requirement that the messages sent by each party be private. While protocols in the central model generally allow better accuracy, they require a much greater level of trust to be placed in the analyzer, an assumption that may be unsuitable for certain applications. The \emph{shuffled model} is based on the Encode-Shuffle-Analyze architecture of \cite{bittau17} and was first analytically studied by~\cite{erlingsson2019amplification,cheu19} and further studied in recent work~\cite{balle_privacy_2019,ghazi2019private}. It seeks to bridge the two aforementioned models and assumes the presence of a trusted shuffler that randomly permutes all incoming messages from the parties before passing them to the analyzer (see Section~\ref{sec:prelim} for formal definitions.) The shuffled model is particularly compelling because it allows the possibility of obtaining more accurate communication-efficient protocols than in the local model while placing far less trust in the analyzer than in the central model. Indeed, the power of the shuffled model has been illustrated by a number of recent works that have designed algorithms in this model for a wide range of problems such as privacy amplification, histograms, heavy hitters, and range queries~\cite{cheu19,erlingsson2019amplification,balle_privacy_2019,ghazi2019private}.

The appeal of the shuffled model provides the basis for our study of differentially private protocols for aggregation in this work. Most relevant to the present work are the recent differentially private protocols for aggregation of real numbers in the shuffled model provided by \cite{cheu19,balle_privacy_2019,GPV19,BBGN19}.
The strongest of these results~\cite{BBGN19} shows that an extension of the split and mix protocol yields an $(\varepsilon, \delta)$-differentially private protocol for aggregation with error $O(1+1/\varepsilon)$ and $m = O(\log(n/\delta))$ messages, each consisting of $O(\log n)$ bits.

\subsection{Our Results}
\label{subsec:our-results}

\paragraph{Upper bound.}
We prove that the split and mix protocol is in fact secure for a much smaller number of messages. 
In particular, for the same security parameter $\sigma$, the number of messages required in our analysis is %a multiplicative factor of
$\Theta(\log n)$ times smaller than the bound in~\cite{IKOS06}:

\begin{theorem}[Improved upper bound for split and mix]\label{th:up_bd_sec}
Let $n$ and $q$ be positive integers and $\sigma$ be a positive real number. The split and mix protocol (Algorithm~\ref{alg:ishai_et_al} and~\cite{IKOS06}) with $n$ parties and inputs in $\inputring$ is $\sigma$-secure for $m$ messages, where $m = O\left(1 + \frac{\sigma + \log q}{\log n}\right)$ .
\end{theorem}

An interesting case to keep in mind is when the field size $q$ and the inverse statistical distance $2^{\sigma}$ are bounded by a polynomial in $n$. In this case, Theorem~\ref{th:up_bd_sec} implies that the protocol works already with a \emph{constant} number of messages, improving upon the known $O(\log n)$ bound.

\medskip

\paragraph{Lower bound.}
We show that, in terms of the number of messages $m$ sent by each party, Theorem~\ref{th:up_bd_sec} is essentially tight not only for just the split and mix protocol but also for \emph{every} one-round protocol.

\begin{theorem}[Lower bound for every one-round protocol]\label{th:message_lb}
Let $n$ and $q$ be positive integers, and $\sigma\geq 1$ be a real number.
In any $\sigma$-secure, one-round aggregation protocol over $\inputring$ in the anonymized model, each of the $n$ parties must send  
$\Omega\left(1+\frac{\sigma}{\log(\sigma n)} +\frac{\log q}{\log n} \right)$ 
messages.
\end{theorem}
The lower bound holds regardless of the message size and asymptotically matches the upper bound under the very mild assumption that $\sigma$ is bounded by a polynomial in $n$.
%The requirement $\sigma \leq n^{1000}$ may be replaced by $\sigma \leq n^c$ for any constant $c$. This is a very mild condition; moreover, 
Furthermore, when $\sigma$ 
is larger,
% exceeds this threshold
%, we still have a lower bound of the form $\Omega\left(\frac{\sigma}{\log \sigma} + \frac{\log q}{\log n}\right)$, which is 
the bound is
tight up to a factor $O\left(\frac{\log \sigma}{\log n}\right)$.

We point out that Theorem~\ref{th:message_lb} provides a nearly-tight lower bound on the \emph{number of messages}. In terms of the total communication per party, improvements are still possible when $\sigma + \log q = \omega(\log n)$. We discuss this further, along with other interesting open questions, in Section~\ref{sec:conclusion}.

% \medskip

\paragraph{Corollary for Differentially Private Aggregation.}

As stated earlier, the differentially private aggregation protocols of~\cite{BBGN19,GPV19} both use extensions of the split and mix protocol. Moreover, Balle et al.~use the security guarantee of the split and mix protocol as a blackbox and derive a differential privacy guarantee from it~\cite[Lemma~4.1]{BBGN19}.
Specifically, when~$\varepsilon$ is a constant and $\delta \geq \frac{1}{\poly(n)}$, their proof uses the split and mix protocol with field size $q = \poly(n)$. 
Previous analyses required $m = \Omega(\log n)$; however, our analysis works with a \emph{constant} number of messages. 
In general, Theorem~\ref{th:up_bd_sec} implies $(\varepsilon, \delta)$-differential privacy with a factor $\Theta(\log n)$ fewer messages than known before:

\begin{corollary}[Differentially private aggregation in the shuffled model]\label{cor:up_bd_DP}
Let $n$ be a positive integer, and let $\epsilon$, $\delta$ be positive real numbers.
%and $\delta \in [0,1]$. Definitely delta>0, but why require delta <= 1?
There is an $(\epsilon, \delta)$-differentially private aggregation protocol in the shuffled model for inputs in $[0,1]$ having absolute error 
%$O_{\epsilon}(1)$. 
$O(1 + 1/\varepsilon)$ in expectation, using $O\left(1 + \frac{\epsilon + \log(1/\delta)}{\log n}\right)$ messages per party, each consisting of $O(\log n)$ bits.
\end{corollary}

A more comprehensive comparison between our differentially private aggregation protocol in Corollary~\ref{cor:up_bd_DP} and previous protocols is presented in Figure~\ref{fig:comparison}.

% Under the mild condition that $\varepsilon > 2\ln(n)/n$ the expected error is constant.

We end this subsection by remarking that Ishai et al.~\cite{IKOS06} in fact considered a setting that is more general than what we have described so far. Specifically, they allow the adversary to corrupt a certain number of parties. 
In addition to the transcript of the protocol, the adversary knows the input and messages of these corrupted parties. (Alternatively, one can think of these corrupted parties as if they are colluding to learn the information about the remaining parties.) 
As already observed in~\cite{IKOS06}, the security of the split and mix protocol still holds in this setting except that $n$ is now the number of honest (i.e., uncorrupted) parties. 
In other words, Theorem~\ref{th:up_bd_sec} remains true in this more general setup but with $n$ being the number of honest parties instead of the total number of parties.

\paragraph{Discussion and comparison of parallel work.}
Concurrently and independently of our work, Balle et al.~\cite{balle_privacy_2019constantIKOS} obtained an upper bound that is asymptotically the same as the one in Theorem~\ref{th:up_bd_sec}. They also give explicit constants, whereas we state our theorem in asymptotic notation and do not attempt to optimize the constants in our proof.

A key difference between our work and theirs is that in addition to the analysis of the split and mix protocol, we manage to prove a matching lower bound on the required number of messages for any protocol (see Theorem~\ref{th:message_lb}), which establishes the near-tightness of the algorithmic guarantees in our upper bound. Our lower bound approach could potentially be applied to other problems pertaining to the anonymous model and possibly differential privacy.

The upper bound proofs use different techniques. Balle et al.~reduce the question to an analysis of the number of connected components of a certain random graph, while our proof analyzes the rank deficiency of a carefully-constructed random matrix. While the upper bound of Balle et al.~is shown for summation over any abelian group, our proofs are presented for finite fields. We note, though, that our lower bound proof carries over verbatim to any abelian group.

\subsection{Applications and Related Work}

At first glance it may seem that aggregation is a rather limited primitive for combining data from many sources in order to analyze it.
However, in important approaches to machine learning and distributed/parallel data processing, the mechanism for combining computations of different parties is \emph{aggregation of vectors}.
Since we can build vector aggregation in a straightforward way from scalar aggregation, our results can be applied in these settings.

Before discussing this in more detail, we mention that it is shown in~\cite{IKOS06} that summation protocols can be used as building blocks for realizing \emph{general} secure computations in a specific setup where a server mediates computation of a function on data held by $n$ other parties.
However, the result assumes a somewhat weak security model (see in Appendix D of \cite{IKOS06} for more details).

\paragraph{Machine Learning.}
Secure aggregation has applications in so-called \emph{federated} machine learning \cite{konevcny2016federated}.
The idea is to train a machine learning model without collecting data from any party, and instead compute weight updates in a distributed manner by sending model parameters to all parties, locally running stochastic gradient descent on private data, and aggregating model updates over all parties. 
For learning algorithms based on gradient descent, a secure aggregation primitive can be used to compute global weight updates without compromising privacy~\cite{mcmahan2016communication,GoogleBlog17}.
%A federated learning system based on these principles is currently used by Google to train neural networks on data residing on parties' phones~\cite{GoogleBlog17}.
It is known that gradient descent can work well even if data is accessible only in noised form, in order to achieve differential privacy~\cite{abadi2016deep}.
%Note that in order to run gradient descent in a differentially private manner, privacy parameters need to be chosen in such a way that the combined privacy loss over many iterations is limited.

Beyond gradient descent, as observed in~\cite{cheu19}, we can translate any \emph{statistical query} over a distributed data set to an aggregation problem over numbers in $[0,1]$.
That is, every learning problem solvable using a small number of statistical queries~\cite{kearns1998efficient} can be solved privately and efficiently based on secure aggregation.

\paragraph{Sketching.}
Research in the area of data stream algorithms has uncovered many non-trivial algorithms that are compact \emph{linear sketches}, see, e.g.,~\cite{cormode2011synopses,woodruff2014sketching}.
As noted already in~\cite{IKOS06}, linear sketches can be implemented using secure aggregation by computing linear sketches locally, and then using aggregation to compute their sum which yields the sketch of the whole dataset.
Typically, linear sketches do not reveal much information information about their input, and are robust to the noise needed to ensure differential privacy, though specific guarantees depend on the sketch in question.
We refer to~\cite{corr/abs-1204-2606,MelisDC16,mishra2006privacy} for examples and further discussion.

%This unlocks many differentially private protocols in the shuffled model, e.g.~estimation of $\ell_p$-norms, quantiles, heavy hitters, and number of distinct elements.

\paragraph{Secure aggregation protocols.}
Secure aggregation protocols are well-studied, both under cryptographic assumptions and with respect to differential privacy.
We refer to the survey of Goryczka et al.~\cite{secAggSurvey} for an overview, but note that our approach leads to protocols that use less communication than existing (multi-round) protocols.
The trust assumptions needed for implementing a shuffler (e.g., using a mixnet) are, however, slightly different from the assumptions typically used for secure aggregation protocols.
Practical secure aggregation typically relies on an honest-but-curious assumption, see e.g.~\cite{practicalSecAgg}. In that setting, such protocols typically require five rounds of communication with $\Omega(n)$ bits of communication and $\Omega(n^2)$ computation per party. A more recent work \cite{reyzin2018turning} using homomorphic threshold encryption gives a protocol with three messages and constant communication and computation per party in addition to a (reusable) two-message setup (consisting of $\Omega(n)$ communication per party). By contrast, our aggregation protocol has a single round of constant communication and computation per party, albeit in the presence of a trusted shuffler.
%Different variations of secure aggregation protocols have also been suggested and compared in~\cite{practicalSecAgg}. 

\paragraph{Other related models.}

A very recent work \cite{wang2019practical} has designed an extension of the shuffled model, called \emph{Multi Uniform Random Shufflers} and analyzed its trust model and privacy-utility tradeoffs. Since they consider a more general model, our differentially private aggregation protocol would hold in their setup as well.

There has also been work on aggregation protocols in the multiple servers setting, e.g., the PRIO system~\cite{corrigan2017prio}; here the protocol is secure as long as at least one server is honest.
Thus trust assumptions of PRIO are somewhat different from those underlying shuffling and mixnets.
While each party would be able to check the output of a shuffler, to see if its message is present, such a check is not possible in the PRIO protocol making server manipulation invisible even if the number of parties is known.
On the other hand, PRIO handles malicious parties that try to manipulate the result of a summation by submitting illegal data --- a challenge that has not been addressed yet for summation in the Shuffled model but that would be interesting future work.

\begin{figure}[t]
    \centering
    % \bgroup
    \def\arraystretch{1.5}
    \begin{tabular}{|l|c|c|c|}
        \hline
        {\bf Reference} & 
        {\bf \#messages / $n$} & 
        {\bf Message size} & 
        {\bf Expected error}\\
        \hline
        \hline
        Cheu et al.~\cite{cheu19} & 
        \begin{tabular}{@{}c@{}} $\varepsilon\sqrt{n}$\\ $\ell$\end{tabular} &
        1 & 
        \begin{tabular}{@{}c@{}} $\frac{1}{\varepsilon} \log\frac{n}{\delta}$\\ $\sqrt{n} / \ell + \frac{1}{\varepsilon} \log\frac{1}{\delta}$ \end{tabular}\\
        \hline
        Balle et al.~\cite{balle_privacy_2019} & $1$ & $\log n$ & $\frac{n^{1/6}\log^{1/3}(1/\delta)}{\varepsilon^{2/3}}$\\
        \hline
        Ghazi et al~\cite{GPV19} &
        $\log(\tfrac{n}{\varepsilon \delta})$ & $\log(\tfrac{n}{\delta})$ & $\frac{1}{\varepsilon} \sqrt{\log\frac{1}{\delta}}$\\ 
        \hline
        Balle et al.~\cite{BBGN19} &
        $\log(\tfrac{n}{\delta})$ & $\log n$ & $\frac{1}{\varepsilon}$\\         
        \hline 
        \emph{This work (Corollary~\ref{cor:up_bd_DP})} & $1 + \frac{\log(1/\delta)}{\log n}$ & $\log n$ & $\frac{1}{\varepsilon}$ \\
        \hline
    \end{tabular}
    \caption{Comparison of differentially private aggregation protocols in the shuffled model with $(\varepsilon,\delta)$-differential privacy. 
    The number of parties is $n$, and $\ell$ is an integer parameter.
    Message sizes are in bits. For readability, we assume that $\varepsilon \leq O(1)$, and asymptotic notations are suppressed.
    }
    \label{fig:comparison}
\end{figure}

\subsection{The Split and Mix Protocol}

The protocol of~\cite{IKOS06} is shown in Algorithm~\ref{alg:ishai_et_al}.
To describe the main guarantee proved in~\cite{IKOS06} regarding Algorithm~\ref{alg:ishai_et_al}, we need some notation. 
For any input sequence $\bx \in \mathbb{F}_q^n$, we denote by $\encshufdist{\bx}{}$ the distribution on $\F_q^{mn}$ obtained by sampling $y_{m(i - 1) + 1}, \dots, y_{mi} \in \F_q$ uniformly at random conditioned on $y_{m(i - 1) + 1} + \cdots + y_{mi} = x_i$, sampling a random permutation $\pi: [mn] \to [mn]$, and outputting $(y_{\pi(1)}, \dots, y_{\pi(mn)})$. 
Ishai et al.~\cite{IKOS06} proved that for some $m = O(\log{n} + \sigma + \log{q})$ and for any two input sequences $\bx, \bx' \in \mathbb{F}_q^n$ having the same sum (in $\mathbb{F}_q$), the distributions $\encshufdist{\bx}{}$ and $\encshufdist{\bx'}{}$ are $2^{-\sigma}$-close in statistical distance.

\begin{algorithm}[h]
\KwIn{$x \in \mathbb{F}_q$, positive integer parameter $m$}
\KwOut{Multiset $\{y_1,\dots, y_m\} \subseteq \mathbb{F}_q$} \vspace{5pt}

\For{$j = 1,\dots,m-1$}{
	${y}_j \leftarrow \uniform(\mathbb{F}_q)$
}
${y}_m \leftarrow x - \sum_{j=1}^{m-1} {y}_j$ (in $\mathbb{F}_q$)\\
\Return{$\{{y}_1,\dots, {y}_m\}$}
\caption{Split and mix encoder from~\cite{IKOS06}}%Encoder of Ishai, Kushilevitz, Ostrovsky and Sahai \cite{IKOS06}}
\label{alg:ishai_et_al}
\end{algorithm}

\subsection{Overview of Proofs}
\label{subsec:overview}
We now give a short overview of the proofs of Theorems~\ref{th:up_bd_sec} and \ref{th:message_lb}. For ease of notation, we define $\sumset_s$ to be the set of all input vectors $\bx = (x_1, x_2, \dots, x_n) \in \inputring^n$ with a fixed sum $x_1 + x_2 + \cdots + x_n = s$.

\paragraph{Upper Bound.}
To describe the main idea behind our upper bound, we start with the following notation. For every $x \in \mathbb{F}_q$, we denote by $\encshufdist{x}{}$ the uniform distribution on $\F^{mn}_q$ conditioned on all coordinates summing to $x$.

To prove Theorem~\ref{th:up_bd_sec}, we have to show that for any two input sequences $\bx, \bx' \in \F_q^n$ such that $\sum_{i \in [n]} x_i = \sum_{i \in [n]} x'_i$, the statistical distance between $\encshufdist{\bx}{}$ and $\encshufdist{\bx'}{}$ is at most $\gamma = 2^{-\sigma}$. By the triangle inequality, it suffices to show that the statistical distance between $\encshufdist{\bx}{}$ and $\encshufdist{x_1 + \cdots + x_n}{}$ is at most $\gamma/2$. (Theorem~\ref{thm:main-central}). Note that $\encshufdist{x_1 + \cdots + x_n}{}$ puts equal mass on all vectors in $\F^{mn}_q$ whose sum is equal to $x_1 + \cdots + x_n$. Thus, our task boils down to showing that the mass put by $\encshufdist{\bx}{}$ on a random sample from $\encshufdist{x_1 + \cdots + x_n}{}$ is well-concentrated. We prove this via a second order method (specifically, Chebyshev's inequality). This amounts to computing the mean and bounding the variance. The former is a simple calculation whereas the latter is more technically involved and reduces to proving a probabilistic bound on the rank deficit of a certain random matrix (Theorem~\ref{thm:corank-bound}). A main ingredient in the proof of this bound is a combinatorial characterization of the rank deficit of the relevant matrices in terms of \emph{matching partitions} (Lemma~\ref{lem:corank-condition}).

\paragraph{Lower Bound.}
For the lower bound (Theorem~\ref{th:message_lb}), our proof consists of two parts: a ``security-dependent'' lower bound $m \geq \Omega\left(\frac{\sigma}{\log(\sigma n)}\right)$ and a ``field-dependent'' lower bound $m \geq \Omega\left(\frac{\log q}{\log n}\right)$. Combining these two yields Theorem~\ref{th:message_lb}. We start by outlining the field-dependent bound as it is simpler before we outline the security-dependent lower bound which is technically more challenging.

\paragraph{Field-Dependent Lower Bound.}

To prove the field-dependent lower bound (formally stated in Theorem~\ref{thm:lognq-lowerbound}), the key idea is to show that for any $s\in\inputring$, there exist distinct inputs $\bx, \bx' \in\sumset_s$ such that the statistical distance between $\encshufdist{\bx}{}$ and $\encshufdist{\bx'}{}$ is at least $1- n^{nm}/q^{n-1}$ (see Lemma~\ref{lem:nqtech}). We do so by proving the same quantitative lower bound on the \emph{average} statistical distance between $\encshufdist{\bx}{}$ and $\encshufdist{\bx'}{}$ over all pairs $\bx, \bx' \in \sumset_s$.

The average statistical distance described above can be written as the sum, over all $\bf{y}$, of the average difference in probability mass assigned to $\bf{y}$ by $\bx$ and $\bx'$. Thus, we consider how to lower bound this coordinate-wise probability mass difference for an arbitrary $\bf{y}$.

There are at most $n^{nm}$ ways to associate each of the $nm$ elements of $\by$ with a particular party. Since any individual party's encoding uniquely determines the corresponding input, it follows that any shuffled output $\by$ could have arisen from at most $n^{nm}$ inputs $\bx$. Moreover, since there are exactly $q^{n-1}$ input vectors $\bx\in\sumset_s$, it follows that there are at least $q^{n-1} - n^{nm}$ possible inputs $\bx\in\sumset_s$ that cannot possibly result in $\by$ as an output. This implies that the average coordinate-wise probability mass difference, over all $\bx, \bx' \in \sumset_s$, is at least $\left(1 - \frac{q^{n-1}}{n^{nm}}\right)$ times the average probability mass assigned to $\by$ over all inputs in $\sumset_s$. Summing this up over all $\by$ yields the desired bound.

\paragraph{Security-Dependent Lower Bound.}

To prove the security-dependent lower bound, it suffices to prove the following statement (see Theorem~\ref{thm:lower-bound-n-to-m}): if $\enc$ is the encoder of any aggregation protocol in the anonymized model for $n > 2$ parties with $m$ messages sent per party, then there is a vector $\bx \in \sumset_0$ such that the statistical distance between the distributions of the shuffled output $\by$ corresponding to inputs $\bzero$ and $\bx$ is at least $\frac{1}{(10nm)^{5m}}$. %It is simple to see that this implies that $m \geq \Omega(\frac{\sigma}{\log(\sigma n)})$ as desired.

Let us first sketch a proof for the particular case of the split and mix protocol. In this case, we set $\bx = (\underbrace{1,1, \dots, 1}_{n-1}, -(n - 1))$, and we will bound from below the statistical distance by considering the ``distinguisher'' $\cA$ which chooses a random permutation $\pi: [nm]\to[nm]$ and accepts iff $y_{\pi(1)} + \cdots + y_{\pi(m)} = 0$. We can argue (see Subsection~\ref{subsec:sec_dep_lb}) that the probability that $\cA$ accepts under the distribution $\encshufdist{\bzero}{}$ is larger by an additive factor of $\frac{1}{(en)^m}$ than the probability that it accepts under the distribution $\encshufdist{\bx}{}$. To generalize this idea to arbitrary encoders (beyond Ishai et al.'s protocol), it is natural to consider a distinguisher which accepts iff $y_{\pi(1)}, \dots, y_{\pi(m)}$ is a valid output of the encoder when the input is zero. Unlike the case of Ishai et al., in general when $\pi(1), \dots, \pi(m)$ do not all come from the same party, it is not necessarily true that the acceptance probability would be the same for both distributions. To circumvent this, we pick the smallest integer $t$ such that the $t$-message marginal of the encoding of 0 and that of input 1 are substantially different, and we let the distinguisher perform an analogous check on $y_{\pi(1)}, \dots, y_{\pi(t)}$ (instead of $y_{\pi(1)}, \dots, y_{\pi(m)}$ as before). Another complication that we have to deal with is that we can no longer consider the input vector $(1, \cdots, 1, -(n - 1))$ as in the lower bound for Ishai et al.'s protocol sketched above. This is because the $t$-message marginal of the encoding of $-(n - 1)$ could deviate from that for input $0$ more substantially than from that for input 1, which could significantly affect the acceptance probability. Hence, to overcome this issue, we instead set $x^*$ to the minimizer of this value $t$ among all elements of $\F_q$, and use the input vector $\bx = (x^*, \dots, x^*,- (n - 1)x^*)$ (for more details we refer the reader to the full proof in Subsection~\ref{subsec:sec_dep_lb}).

\subsection*{Organization of the Rest of the Paper}
We start with some preliminaries in Section~\ref{sec:prelim}. We prove our main upper bound (Theorem~\ref{th:up_bd_sec}) in Section~\ref{sec:ub_pf}. We prove our lower bound (Theorem~\ref{th:message_lb}) in Section~\ref{sec:lb_pf}. The proof of Corollary~\ref{cor:up_bd_DP} appears in Appendix~\ref{sec:cor_pf}.

\section{Preliminaries}\label{sec:prelim}
\subsection{Protocols}
In this paper, we are concerned with answering the question of how many messages are needed for protocols to achieve certain security or cryptographic guarantees. We formally define the notion of protocols in the models of interest to us.

We first define the notion of a \emph{secure protocol} in the \emph{shuffled model}. An $n$-user \emph{secure protocol} in the \emph{shuffled model}, $\cP = (\enc, \analyzer)$, consists of a randomized \emph{encoder} (also known as \emph{local randomizer}) $\enc: \cX \to \cY^m$ and an \emph{analyzer} $\analyzer: \cY^{nm} \to \cZ$. Here, $\cY$ is known as the \emph{message alphabet}, $\cY^m$ is the \emph{message space} for each user, and $\cZ$ is the \emph{output space} of the protocol. The protocol $\cP$ implements the following mechanism: each party $i$ holds an input $x_i\in\cX$ and encodes $x_i$ as $\enc_{x_i}$. (Note that $\enc_{x_i}$ is possibly random based on the private randomness of party $i$.) The concatenation of the encodings, $\by = (\enc_{x_1}, \enc_{x_2}, \dots, \enc_{x_n}) \in \cY^{nm}$ is then passed to a trusted \emph{shuffler}, who chooses a uniformly random permutation $\pi$ on $nm$ elements and applies $\pi$ to $\by$. The output is submitted to the analyzer, which then outputs $\cP(\bx) = \analyzer(\pi(\by)) \in \cZ$.

In this paper, we will be concerned with protocols for \emph{aggregation}, in which $\cX = \cZ = \inputring$ (a finite field on $q$ elements) and $\cY = [\ell] = \{1,2,\dots, \ell\}$, and
\[
\analyzer(\pi(\enc_{x_1}, \enc_{x_2}, \dots, \enc_{x_n})) = \sum_{i=1}^n x_i,
\]
i.e., the protocol always outputs the sum of the parties' inputs, regardless of the randomness over the encoder and the shuffler.

A related notion that we consider in this work is a one-round protocol $\cP = (\enc, \analyzer)$ in the \emph{anonymized model}. The notion is similar to that of a secure protocol in the shuffled model except that there is no shuffler. Rather, the analyzer $\analyzer$ receives a \emph{multiset} of $nm$ messages obtained by enumerating all $m$ messages of each of the $n$ parties' encodings. It is straightforward to see that the two models are equivalent, in the sense that a protocol in one model works in the other and the distributions of the view of the analyzer are the same.

\subsection{Distributions Related to a Protocol}

%We now define some notation for natural probability distributions resulting from the application of a protocol for aggregation.

To study a protocol and determine its security and privacy, it is convenient to define notations for several probability distributions related to the protocol. First, we use $\enclocal{x}{\enc}$ to denote the distribution of the (random) encoding of $x$:

\begin{definition}
For a protocol $\cP$ with encoding function $\enc$, we let $\enclocal{x}{\enc}$ denote the distribution of outputs over $\cY^m$ obtained by applying $\enc$ to $x\in\cX$.
\end{definition}

Furthermore, for a vector $\bx \in \cX^n$, we use $\encdist{\bx}{\enc}$ to denote the distribution of the concatenation of encodings of $x_1, \dots, x_n$, as stated more formally below.

\begin{definition}
For an $n$-party protocol $\cP$ with encoding function $\enc$ and $\bx\in\cX^n$, we let $\encdist{\bx}{\enc}$ denote the distribution over $\cY^{nm}$ obtained by applying $\enc$ individually to each element of $\bx$, i.e.,
\[
 \encdist{\bx}{\enc} \sim \left(\enclocal{x_1}{\enc}, \enclocal{x_2}{\enc}, \dots, \enclocal{x_n}{\enc}\right).
\]
\end{definition}

Finally, we define $\encshufdist{\bx}{\enc}$ to be $\encdist{\bx}{\enc}$ after random shuffling. Notice that $\encshufdist{\bx}{\enc}$ is the distribution of the transcript seen at the analyzer.

\begin{definition}
For an $n$-party protocol $\cP$ with encoding function $\enc$ and $\bx\in\cX^n$, we let $\encshufdist{\bx}{\enc}$ denote the distribution over $\cY^{nm}$ obtained by applying $\enc$ to the elements of $\bx$ and then shuffling the resulting $nm$-tuple, i.e.,
\[
 \encshufdist{\bx}{\enc} \sim \pi \circ \encdist{\bx}{\enc}
\]
for $\pi$ a uniformly random permuation over $nm$ elements.
\end{definition}

\subsection{Security and Privacy}
Given two distributions $\cD_1$ and $\cD_2$, we let $\sd(\cD_1, \cD_2)$ denote the \emph{statistical distance} (aka the total variation distance) between $\cD_1$ and $\cD_2$.

We begin with a notion of $\sigma$-security for computation of a function $f$, which essentially says that distinct inputs with a common function value should be (almost) indistinguishable:
\begin{definition}[$\sigma$-security]
An $n$-user one-round protocol $\cP = (\enc, \cA)$ in the anonymized model is said to be $\sigma$-secure for computing a function $f:\cX^n \to \cZ$ if for any $\bx, \bx' \in \cX^n$ such that $f(\bx)=f(\bx')$, we have
\[
  \sd\left(\encshufdist{\bx}{\enc}, \encshufdist{\bx'}{\enc}\right) \leq 2^{-\sigma}.
\]
\end{definition}
In this paper, we will primarily be concerned with the function that sums the inputs of each party, i.e., $f:\inputring^n \to \inputring$ given by $f(x_1, x_2, \dots, x_n) = \sum_{i=1}^n x_i$.

We now define the notion of \emph{$(\varepsilon,\delta)$-differential privacy}. We say that two input vectors $\bx = (x_1,x_2, \dots, x_n) \in \cX^n$ and $\bx' = (x_1', x_2', \dots, x_n') \in \cX^n$ are \emph{neighboring} if they differ on at most one party's data, i.e., $x_i = x_i'$ for all but one value of $i$.
\begin{definition}[$(\varepsilon, \delta)$-differential privacy]
 An algorithm $M: \cX^* \to \cZ$ is $(\varepsilon,\delta)$-differentially private if for every neighboring input vectors $\bx, \bx' \in \cX^n$ and every $S \subseteq \cZ$, we have
 \[
   \Pr[M(\bx)\in S] \leq e^\varepsilon \cdot \Pr[M(\bx')\in S] + \delta,
 \]
 where probability is over the randomness of $M$.
\end{definition}

We now define $(\epsilon,\delta)$-differential privacy specifically in the \emph{shuffled model}.
\begin{definition}
 A protocol $\cP$ with encoder $\enc: \cX\to\cZ^m$ is $(\varepsilon, \delta)$-differentially private in the shuffled model if the algorithm $M: \cX^n \to \cZ^{nm}$ given by
 \[
   M(x_1, x_2, \dots, x_n) = \pi(\enc_{x_1}, \enc_{x_2}, \dots, \enc_{x_n})
 \]
 is $(\epsilon,\delta)$-differentially private, where $\pi$ is a uniformly random permutation on $nm$ elements.
\end{definition}

\section{Proof of Theorem~\ref{th:up_bd_sec}}\label{sec:ub_pf}

In this section, we prove Theorem~\ref{th:up_bd_sec}, i.e., that the split and mix protocol of Ishai et al. is $\sigma$-secure even for $m = \Theta\left(1 + \frac{\sigma + \log q}{\log n}\right)$ messages, improving upon the known bounds of $O(\log{n} + \sigma + \log{q})$~\cite{IKOS06,BBGN19,GPV19}.

Since we only consider Ishai et al.'s split and mix protocol in this section, we will drop the superscript from $\encshufdist{\bx}{\enc}$ and simply write $\encshufdist{\bx}{}$ to refer to the shuffled output distribution of the protocol. Recall that, by the definition of the protocol, $\encshufdist{\bx}{}$ is generated as follows: for every $i \in [n]$, sample $y_{m(i - 1) + 1}, \dots, y_{mi} \in \F_q$ uniformly at random conditioned on $y_{m(i - 1) + 1} + \cdots + y_{mi} = x_i$. Then, pick a random permutation $\pi: [mn] \to [mn]$ and output $(y_{\pi(1)}, \dots, y_{\pi(mn)})$.

Showing that the protocol is $\sigma$-secure is by definition equivalent to showing that $\sd(\encshufdist{\bx}{}, \encshufdist{\bx'}{}) \leq 2^{-\sigma}$ for all inputs $\bx, \bx' \in \F_q^n$ such that $\sum_{i \in [n]} x_i = \sum_{i \in [n]} x_i$.

%\begin{theorem} \label{thm:main-pairwise}
%Let $n, m, q$ be positive integers, such that $q$ is a prime power. For every $\bx = (x_1, \dots, x_n) \in \F_q^n$, let $\encshufdist{\bx}{}$ denote the distribution on $\F_q^{mn}$ generated as follows: for every $i \in [n]$, sample $y_{m(i - 1) + 1}, \dots, y_{mi} \in \F_q$ uniformly at random conditioned on $y_{m(i - 1) + 1} + \cdots + y_{mi} = x_i$. Then, pick a random permutation $\pi: [mn] \to [mn]$ and output $(y_{\pi(1)}, \dots, y_{\pi(mn)})$. For any parameter $\gamma > 0$ and any $m \geq \Theta(1+\log_n(q/\gamma))$, the following holds: for every $\bx, \bx' \in \F_q^n$ such that $\sum_{i \in [n]} x_i = \sum_{i \in [n]} x'_i$, the statistical distance between $\encshufdist{\bx}{}$ and $\encshufdist{\bx'}{}$ is at most $\gamma$. 
%\end{theorem}

%Specifically, the above result gives a quantitative improvement over previous results from~\cite{IKOS06,GPV19,BBGN19}, which requires $m$ to be at least $\Theta(\log(q/\gamma))$, instead of $\Theta(\log_n(q/\gamma))$ as in our theorem. In particular, when $q \geq n^{\Omega(1)}$ or $\gamma \leq 1/n^{\Omega(1)}$, our bound yields a saving of order $\log n$ in terms of the number of messages $m$.

In fact, we prove a stronger statement, that each $\encshufdist{\bx}{}$ is $\gamma$-close (in statistical distance) to the distribution that is uniform over all vectors in $\F_q^{m n}$ whose sum of all coordinates is equal to $\sum_{i \in [n]} x_i$, as stated below.

\begin{theorem} \label{thm:main-central}
For every $a \in \F_q$, let $\encshufdist{a}{}$ denote the distribution on $\F^{mn}_q$ generated uniformly at random conditioned on all coordinates summing to $a$. For any parameter $\gamma > 0$ and any $m \geq \Theta(1 + \log_n(q/\gamma))$, the following holds: for every $\bx \in \F_q^n$, the statistical distance between $\encshufdist{\bx}{}$ and $\encshufdist{x_1 + \cdots + x_n}{}$ is at most $\gamma$.
\end{theorem}

When plugging in $\gamma = 2^{-\sigma - 1}$, Theorem~\ref{thm:main-central} immediately implies Theorem~\ref{th:up_bd_sec} via the triangle inequality.

We now outline the overall proof approach. First, observe that $\encshufdist{x_1 + \cdots + x_n}{}$ puts probability mass equally across all vectors $\bt \in \F_q^{mn}$ whose sum of all coordinates is $x_1 + \cdots + x_n$, whereas $\encshufdist{\bx}{}$ puts mass proportional to the number of permutations $\pi:[mn] \to [mn]$ such that $\by := (t_{\pi^{-1}(1)}, \dots, t_{\pi^{-1}(mn)})$ satisfies $y_{m(i - 1) + 1} + \cdots + y_{mi} = x_i$ for all $i \in [n]$. Thus, our task boils down to proving that this latter number of is well-concentrated (for a random $\bt \in \supp(\encshufdist{x_1 + \cdots + x_n}{})$). We prove this via a second moment method (specifically Chebyshev's inequality). Carrying this out amounts to computing the first moment and upper-bounding the second moment of this number. The former is a simple calculation, whereas the latter involves proving an inequality regarding the rank of a certain random matrix (Theorem~\ref{thm:corank-bound}). We do so by providing a combinatorial characterization of the rank deficit of the relevant matrices (Lemma~\ref{lem:corank-condition}).

The rest of this section is organized as follows. In Subsection~\ref{sec:basic-setup-moment}, we define appropriate random variables, state the bound we want for the second moment (Lemma~\ref{eq:prod-bound}), and show how it implies our main theorem (Theorem~\ref{thm:main-central}). Then, in Subsection~\ref{sec:moment-v-rank}, we relate the second moment to the rank of a random matrix (Proposition~\ref{prop:second-moment-v-corank}). Finally, we give a probabilistic bound on the rank of such a random matrix in Subsection~\ref{sec:rank-bound} (Theorem~\ref{thm:corank-bound}).

\subsection{Bounding Statistical Distance via Second Moment Method}
\label{sec:basic-setup-moment}

From now on, let us fix $\bx \in \F_q^n$, and let $a = x_1 + \cdots + x_n$. The variables we define below will depend on $\bx$ (or $a$), but, for notational convenience, we avoid indicating these dependencies in the variables' names.

For every $\bt \in \F_q^{mn}$, let $Z_{\bt}$ denote the number of permutations $\pi: [mn] \to [mn]$ such that $t_{\pi(m(i - 1) + 1)} + \cdots + t_{\pi(mi)} = x_i$ for all $i \in [n]$. From the definition\footnote{Note that, if derived directly from the definition of $\encshufdist{\bx}{}$, $\pi$ here should be replaced by $\pi^{-1}$. However, these two definitions are equivalent since $\pi \mapsto \pi^{-1}$ is a bijection.} of $\encshufdist{\bx}{}$, its probability mass function is
\begin{align}
f_{\encshufdist{\bx}{}}(\bt) = \frac{Z_{\bt}}{(mn)! \cdot q^{(m - 1)n}}.
\end{align}
As stated earlier, Theorem~\ref{thm:main-central} is essentially about the concentration of $Z_{\bt}$, which we will prove via the second moment method. To facilitate the proof, for every $\pi: [mn] \to [mn]$, let us also denote by $Y_{\bt, \pi}$ the indicator variable of ``$t_{\pi(r(i - 1) + 1)} + \cdots + t_{\pi(ri)} = x_i$ for all $i \in [n]$''. Note that by definition we have
\begin{align} \label{eq:sum-over-perm}
Z_{\bt} = \sum_{\pi \in \Pi_{mn}} Y_{\bt, \pi}
\end{align}
where $\Pi_{mn}$ denotes the set of all permutations of $[mn]$.

When we think of $\bt$ as a random variable distributed according to $\encshufdist{a}{}$, the mean of $Y_{\bt, \pi}$ (and hence of $Z_{\bt}$) can be easily computed: the probability that $\bt$ satisfies ``$t_{\pi(m(i - 1) + 1)} + \cdots + t_{\pi(mi)} = x_i$'' is exactly $1/q$ for each $i \in [n - 1]$, and these events are independent. Furthermore, when these events are true, it is automatically the case that the condition holds for $i = n$. Hence, we immediately have:

\begin{observation} \label{obs:exp-y}
For every $\pi \in \Pi_{mn}$,
\begin{align} \label{eq:exp-y}
\E_{\bt \sim \encshufdist{a}{}}\left[Y_{\bt, \pi}\right] = \frac{1}{q^{n - 1}}.
\end{align}
\end{observation}

The more challenging part is upper-bounding the second moment of $Z_{\bt}$ (where we once again think of $\bt$ as a random variable drawn from $\encshufdist{a}{}$). This is equivalent to upper-bounding the expectation of $Y_{\bt, \pi} \cdot Y_{\bt, \pi'}$, where $\pi, \pi'$ are independent uniformly random permutations of $[mn]$ and $\bt$ is once again drawn from $\encshufdist{a}{}$. On this front, we will show the following bound in the next subsections.

\begin{lemma} \label{lem:y-prod-bound}
For every $\pi \in \Pi_{mn}$, we have
\begin{align} \label{eq:prod-bound}
\E_{\pi, \pi' \sim \Pi_{mn}, \bt \sim \encshufdist{a}{}}\left[Y_{\bt, \pi} \cdot Y_{\bt, \pi'}\right] \leq \sum_{k \geq 1} \frac{q^{k}}{q^{2n - 1}} \cdot \left(\frac{n^{2}}{(n/2)^{m - 2}}\right)^{\frac{k - 1}{2}}.
\end{align}
\end{lemma}

Since there are many parameters, the bound might look a bit confusing. However, the only property we need in order to show concentration of $Z_{\bt}$ is that the right-hand side of~\eqref{eq:prod-bound} is dominated by the $k = 1$ term. This is the case when the term inside the parenthesis is $q^{-\Omega(1)}$, which indeed occurs when $m \geq 4 + \Omega(\log_n q)$.

The bound in Lemma~\ref{lem:y-prod-bound} will be proved in the subsequent sections. For now, let us argue why such a bound implies our main theorem (Theorem~\ref{thm:main-central}).

\begin{proof}[Proof of Theorem~\ref{thm:main-central}]
First, notice that~\eqref{eq:sum-over-perm} and Observation~\ref{obs:exp-y} together imply that
\begin{align} \label{eq:exp-z}
\E_{\bt \sim \encshufdist{a}{}}[Z_{\bt}] = \frac{(mn)!}{q^{n - 1}}.
\end{align}
For convenience, let us define $\mu$ as $\frac{(mn)!}{q^{n - 1}}$.

We now bound the second moment of $Z_{\bt}$ as follows:
\begin{align*}
\E_{\bt \sim \encshufdist{a}{}}[Z_{\bt}^2] &= \E_{\bt \sim \encshufdist{a}{}}\left[\left(\sum_{\pi \in \Pi_{mn}} Y_{\bt, \pi}\right)^2\right] \\
&= \left((mn)!\right)^2 \cdot \E_{\pi, \pi' \sim \Pi_{mn}, \bt \sim \encshufdist{a}{}}\left[Y_{\bt, \pi} \cdot Y_{\bt, \pi'}\right] \\
&\overset{~\eqref{eq:prod-bound}}{\leq} \left((mn)!\right)^2 \cdot \left(\sum_{k \geq 1} \frac{q^{k}}{q^{2n - 1}} \cdot \left(\frac{n^{2}}{(n/2)^{m - 2}}\right)^{\frac{k - 1}{2}}\right) \\
&= \left((mn)!\right)^2 \cdot \frac{1}{q^{2(n - 1)}} \cdot \left(1 + \sum_{k \geq 2} q^{k - 1} \cdot \left(\frac{n^{2}}{(n/2)^{m - 2}}\right)^{\frac{k - 1}{2}}\right) \\
&= \mu^2 \cdot \left(1 + \sum_{k \geq 2} \left(\frac{(qn)^{2}}{(n/2)^{m - 2}}\right)^{\frac{k - 1}{2}}\right).
\end{align*}
Now, let $p = \left(\frac{(qn)^{2}}{(n/2)^{m - 2}}\right)^{\frac{1}{2}}$. If $m \geq 4 + 100\log_{n/2}(q/\gamma))$,
then we have $p \leq 0.01\gamma^4$. Plugging this back in the above inequality gives
\begin{align*}
\E_{\bt \sim \encshufdist{a}{}}[Z_{\bt}^2] &\leq \mu^2 \left(\frac{1}{1 - p}\right) \leq \mu^2 \left(\frac{1}{1 - 0.01\gamma^4}\right) \leq \mu^2(1 + 0.02\gamma^4). 
\end{align*}
In other words, we have
\begin{align*}
\Var_{\bt \sim \encshufdist{a}{}}(Z_{\bt}) \leq (0.2\gamma^2 \cdot \mu)^2.
\end{align*}
Hence, by Chebyshev's inequality, we have
\begin{align} \label{eq:chebyshev-result}
\Pr_{\bt \sim \encshufdist{a}{}}[Z_{\bt} \leq (1 - 0.5\gamma) \mu] \leq 0.5\gamma.
\end{align}

Finally, notice that the statistical distance between $\encshufdist{\bx}{}$ and $\encshufdist{a}{}$ is
\begin{align*}
\sum_{\bt \in \F_q^{mn}} \max\{f_{\encshufdist{a}{}}(\bt) - f_{\encshufdist{\bx}{}}(\bt) , 0\} &= \sum_{\substack{\bt \in \F_q^{mn}\\ t_1 + \cdots + t_{mn} = a}} \max\left\{\frac{1}{q^{mn - 1}} - \frac{Z_{\bt}}{(mn)! \cdot q^{(m - 1)n}}, 0\right\} \\
&= \sum_{\substack{\bt \in \F_q^{mn} \\ t_1 + \cdots + t_{mn} = a}} f_{\encshufdist{a}{}}(\bt) \cdot \max\left\{1 - Z_{\bt} / \mu, 0\right\} \\
&= \E_{\bt \sim \encshufdist{a}{}}[\max\left\{1 - Z_{\bt} / \mu, 0\right\}] \\
&\leq \Pr_{\bt \sim \encshufdist{a}{}}[Z_{\bt} \leq (1 - 0.5\gamma)\mu] \cdot 1 + \Pr_{\bt \sim \encshufdist{a}{}}[Z_{\bt} > (1 - 0.5\gamma)\mu] \cdot (0.5 \gamma) \\
&\overset{\eqref{eq:chebyshev-result}}{\leq} (0.5\gamma) \cdot 1 + 1 \cdot (0.5\gamma) \\
&= \gamma. \qedhere
\end{align*}
\end{proof}

\subsection{Relating Moments to Rank of Random Matrices}
\label{sec:moment-v-rank}

Having shown how Lemma~\ref{lem:y-prod-bound} implies our main theorem (Theorem~\ref{thm:main-central}), we now move on to prove Lemma~\ref{lem:y-prod-bound} itself. In this subsection, we deal with the first half of the proof by relating the quantity on the left-hand side of~\eqref{eq:prod-bound} to a quantity involving the rank of a certain random matrix. 

\subsubsection{Warm-Up: (Re-)Computing the First Moment}
\label{subsec:first-moment}

As a first step, let us define below a class of matrices that will be used throughout.

\begin{definition}
For every permutation $\pi: [mn] \to [mn]$, let us denote by $\bA_{\pi} \in \F_q^{n \times mn}$ the matrix whose $i$-th row is the indicator vector for $\pi(\{m(i - 1) + 1, \dots, mi\})$. More formally,
\begin{align*}
(\bA_{\pi})_{i, j} = 
\begin{cases}
1 & \text{ if } j \in \pi(\{m(i - 1) + 1, \dots, mi\}), \\
0 & \text{ otherwise.}
\end{cases}
\end{align*}
\end{definition}

Before we describe how these matrices relate to the second moment, let us illustrate their relation to the first moment, by sketching an alternative way to prove Observation~\ref{obs:exp-y}. To do so, let us rearrange the left-hand side of~\eqref{eq:exp-y} as
\begin{align*}
\E_{\bt \sim \encshufdist{a}{}} [Y_{\bt, \pi}] = \frac{1}{q^{mn - 1}} \sum_{\bt \in \F_q^{mn}} Y_{\bt, \pi}.
\end{align*}
Now, observe that $Y_{\bt, \pi} = 1$ iff $\bA_{\pi} \bt = \bb$. Since the rows of the matrix $\bA_{\pi}$ have pairwise-disjoint supports, the matrix is always full rank (over $\F_q$), i.e., $\rank(\bA_{\pi}) = n$. This means that the number of values of $\bt$ satisfying the aforementioned equation is $q^{mn - n}$. Plugging this into the above expansion gives
\begin{align*}
\E_{\bt \sim \encshufdist{a}{}} [Y_{\bt, \pi}] = \frac{q^{mn - n}}{q^{mn - 1}} = \frac{1}{q^{n - 1}}.
\end{align*}
Hence, we have rederived~\eqref{eq:exp-y}.

\subsubsection{Relating Second Moment to Rank}

%\begin{definition}
%A matrix $\bA \in \F^{m \times n}$ where $m \leq n$ is said to have a \emph{rank deficit} of $n - \rank(\bA)$. Equivalently, the rank deficit of $\bA$ is equal to the corank of $\bA^T$.
%\end{definition}

In the previous subsection, we have seen the relation of matrix $\bA_{\pi}$ to the first moment. We will now state such a relation for the second moment. Specifically, we will rephrase the left-hand side of~\eqref{eq:prod-bound} as a quantity involving matrices $\bA_{\pi}$ and $\bA_{\pi'}$. To do so, we will need the following additional notations:

\begin{definition}
For a pair of permutations $\pi, \pi': [mn] \to [mn]$, we let $\bA_{\pi, \pi'} \in \F_q^{2n \times mn}$ denote the 
(column-wise) concatenation of $\bA_{\pi}$ and $\bA_{\pi'}$, i.e.,
\begin{align*}
\bA_{\pi, \pi'} =
\begin{bmatrix}
\bA_{\pi} \\
\bA_{\pi'}
\end{bmatrix}.
\end{align*}
Furthermore, let\footnote{Note that $\defc(\bA_{\pi, \pi'})$ is equal to the \emph{corank} of $\bA_{\pi, \pi'}^T$.} the \emph{rank deficit} of $\bA_{\pi, \pi'}$ be $\defc(\bA_{\pi, \pi'}) := 2n - \rank(\bA_{\pi, \pi'})$.
\end{definition}

Analogous to the relationship between the first moment and $\bA_{\pi}$ seen in the previous subsection, the quantity $\E_{\bt \sim \encshufdist{a}{}}\left[Y_{\bt, \pi} \cdot Y_{\bt, \pi'}\right]$ is in fact proportional to the number of solutions to certain linear equations, which is represented by $\bA_{\pi, \pi'}$. This allows us to give the bound to the former, as formalized below.

\begin{proposition} \label{prop:second-moment-v-corank}
For every pair of permutations $\pi, \pi': [mn] \to [mn]$, we have
\begin{align*}
\E_{\bt \sim \encshufdist{a}{}}\left[Y_{\bt, \pi} \cdot Y_{\bt, \pi'}\right] \leq \frac{q^{\defc(\bA_{\pi, \pi'})}}{q^{2n - 1}}.
\end{align*}
\end{proposition}

\begin{proof}
First, let us rearrange the left-hand side term as
\begin{align} \label{eq:y-expand}
\E_{\bt \sim \encshufdist{a}{}}\left[Y_{\bt, \pi} \cdot Y_{\bt, \pi'}\right] = \frac{1}{q^{mn - 1}} \sum_{\bt \in \F_q^{mn}} Y_{\bt, \pi} \cdot Y_{\bt \cdot \pi'}.
\end{align}
Now, notice that $Y_{\bt, \pi} = 1$ iff $\bA_{\pi} \bt = \bx$. Similarly, $Y_{\bt, \pi'} = 1$ iff $\bA_{\pi'} \bt = \bx$. In other words, $Y_{\bt, \pi} \cdot Y_{\bt \cdot \pi'} = 1$ iff
\begin{align*}
\bA_{\pi, \pi'} \bt = 
\begin{bmatrix}
\bx \\
\bx
\end{bmatrix}.
\end{align*}
The number of solutions $\bt \in \F_q^{mn}$ to the above equation is at most $q^{mn - \rank(\bA_{\pi, \pi'})} = q^{(m - 2)n + \defc(\bA_{\pi, \pi'}^T)}$. Plugging this back into~\eqref{eq:y-expand}, we get
\begin{align*}
\E_{\bt \sim \encshufdist{a}{}}\left[Y_{\bt, \pi} \cdot Y_{\bt, \pi'}\right] \leq \frac{1}{q^{mn - 1}} \cdot q^{(m - 2)n + \defc(\bA_{\pi, \pi'})} = \frac{q^{\defc(\bA_{\pi, \pi'})}}{q^{2n - 1}},
\end{align*}
as desired.
\end{proof}

\subsection{Probabilistic Bound on Rank Deficit of Random Matrices}
\label{sec:rank-bound}

The final step of our proof is to bound the probability that the rank deficit of $\bA_{\pi, \pi'}$ is large. Such a bound is encapsulated in Theorem~\ref{thm:corank-bound} below. Notice that Proposition~\ref{prop:second-moment-v-corank} and Theorem~\ref{thm:corank-bound} immediately yield Lemma~\ref{lem:y-prod-bound}.

\begin{theorem} \label{thm:corank-bound}
For all $m \geq 3$ and $k \in \N$, we have
\begin{align*}
\Pr_{\pi, \pi' \sim \Pi_{mn}}[\defc(\bA_{\pi, \pi'}) \geq k] \leq \left(\frac{n^{2}}{(n/2)^{m - 2}}\right)^{\frac{k - 1}{2}}.
\end{align*}
\end{theorem}

\subsubsection{Characterization of Rank Deficit via Matching Partitions.}

To prove Theorem~\ref{thm:corank-bound}, we first give a ``compact'' and convenient characterization of the rank deficit of $\bA_{\pi, \pi'}$. In order to do this, we need several additional notations: we say that a partition $S_1 \sqcup \cdots \sqcup S_k = U$ of a universe $U$ is \emph{non-empty} if $S_1, \dots, S_k \ne \emptyset$. Moreover, for a set $S \subseteq [n]$, we use $S^{\to m} \subseteq [mn]$ to denote the set $\cup_{i \in S} \{m(i - 1) + 1, \dots, mi\}$. Finally, we need the following definition of \emph{matching partitions}.

\begin{definition}
Let $\pi, \pi'$ be any pair of permutations of $[mn]$.
A pair of non-empty partitions $S_1 \sqcup \cdots \sqcup S_k = [n]$ and $S'_1 \sqcup \cdots \sqcup S'_k = [n]$ is said to \emph{match with respect to $\pi, \pi'$} iff
\begin{align} \label{eq:corank-condition}
\pi\left(S_j^{\to m}\right) = \pi'\left((S'_j)^{\to m}\right)
\end{align}
for all $j \in [k]$. When $\pi, \pi'$ are clear from the context, we may omit ``with respect to $\pi, \pi'$'' from the terminology.
\end{definition}

Condition~\eqref{eq:corank-condition} might look a bit mysterious at first glance. However, there is a very simple equivalent condition in terms of the matrices $\bA_{\pi}, \bA_{\pi'}$: $S_1 \sqcup \cdots \sqcup S_k = [n]$ and $S'_1 \sqcup \cdots \sqcup S'_k = [n]$ match iff the sum of rows $i \in S_j$ of $\bA_{\pi}$ coincides with the sum of rows $i' \in S'_j$ of $\bA_{\pi'}$, i.e.,  $\sum_{i \in S_j} (\bA_{\pi})_i = \sum_{i' \in S'_j} (\bA_{\pi'})_{i'}$.

An easy-to-use equivalence of $\defc(\bA_{\pi, \pi'}) = k$ is that a pair of matching partitions $S_1 \sqcup \cdots \sqcup S_k = [n]$ and $S'_1 \sqcup \cdots \sqcup S'_k = [n]$ exists. We only use one direction of this relation, which we prove below.

\begin{lemma} \label{lem:corank-condition}
For any permutations $\pi, \pi': [mn] \to [mn]$, if $\defc(\bA_{\pi, \pi'}) \geq k$, then there exists a pair of matching partitions $S_1 \sqcup \cdots \sqcup S_k = [n]$ and $S'_1 \sqcup \cdots \sqcup S'_k = [n]$. 
\end{lemma}

\begin{proof}
We will prove the contrapositive.
Let $\pi, \pi': [mn] \to [mn]$ be any permutations, and suppose that there is no pair of matching partitions $S_1 \sqcup \cdots \sqcup S_k = [n]$ and $S'_1 \sqcup \cdots \sqcup S'_k = [n]$. We will show that $\defc(\bA_{\pi, \pi'}) < k$, or equivalently $\rank(\bA_{\pi, \pi'}) > 2n - k$.

Consider any pair of matching partitions\footnote{Note that at least one matching partition always exists: $S_1 = [n] = S'_1$.} $S_1 \sqcup \cdots \sqcup S_t = [n]$ and $S'_1 \sqcup \cdots \sqcup S'_t = [n]$ that maximizes the number of parts $t$. From our assumption, we must have $t < k$. 

For every part $j \in [t]$, let us pick an arbritrary element $i_j \in S_j$.
Consider all rows of $\bA_{\pi, \pi'}$, except the $i_j$-th rows for all $j \in [t]$ (i.e. $\{(\bA_{\pi, \pi'})_i\}_{i \notin \{i_1, \dots, i_t\}}$). We claim that these rows are linearly independent. Before we prove this, note that this imply that the rank of $\bA_{\pi, \pi'}$ is at least $2n - t > 2n - k$, which would complete our proof.

We now move on to prove the linear independence of $\{(\bA_{\pi, \pi'})_i\}_{i \notin \{i_1, \dots, i_t\}}$. Suppose for the sake of contradiction that these rows are not linearly independent. Since the matrix $\bA_{\pi, \pi'}$ is simply a concatenation of $\bA_{\pi}$ and $\bA_{\pi'}$, we have that $\{(\bA_{\pi, \pi'})_i\}_{i \notin \{i_1, \dots, i_t\}} = \{(\bA_{\pi})_i\}_{i \in [n] \setminus \{i_1, \dots, i_t\}} \cup \{(\bA_{\pi'})_{i'}\}_{i' \in [n]}$. The linear dependency of these rows mean that there exists a non-zero vector of coefficients $(c_1, \dots, c_n, c'_1, \dots, c'_n) \in \F_q^{2n}$ with $c_{i_1} = \cdots = c_{i_t} = 0$ such that 
\begin{align} \label{eq:zero-comb}
\bzero = \sum_{i \in [n]} c_i \cdot (\bA_{\pi})_i + \sum_{i' \in [n]} c'_{i'}\cdot (\bA_{\pi'})_{i'}.
\end{align}
Since the rows of $\bA_{\pi'}$ are linearly independent, there must exist $i^* \in [n]$ such that $c_{i^*} \ne 0$. Let $j \in [t]$ denote the index of the partition to which $i^*$ belongs, i.e., $i^* \in S_j$. For notational convenience, we will assume, without loss of generality, that $j = t$.

Let $P_t: \F_q^{mn} \to \F_q^{(S_t^{\to m})}$ denote the projection operator that sends a vector $(v_\ell)_{\ell \in [mn]}$ to its restriction on coordinates in $S_t^{\to m}$, i.e., $(v_{\ell})_{\ell \in S_t^{\to m}}$. Observe that $P_t((\bA_{\pi})_i)$ is non-zero iff $i \in S_t$ and $P_t((\bA_{\pi'})_{i'})$ is non-zero iff $i' \in S'_t$. Thus, by taking $P_t$ on both sides of~\eqref{eq:zero-comb}, we have
\begin{align} \label{eq:zero-comb-restricted}
\bzero &= \sum_{i \in S_t} c_i \cdot P_t((\bA_{\pi})_i) + \sum_{{i'} \in S'_t} c_{i'} \cdot P_t((\bA_{\pi})_{i'})
\end{align}
Now, let $T = \{i \in S_t \mid c_i \ne 0\}$ and $T' = \{i' \in S'_t \mid c_{i'} \ne 0\}$. Notice that $\supp\left(\sum_{i \in S_t} c_i \cdot P_t((\bA_{\pi})_i)\right) = \pi(T^{\rightarrow m})$ and $\supp\left(\sum_{i' \in S'_t} c_{i'} \cdot P_t((\bA_{\pi})_{i'})\right) = \pi'((T')^{\rightarrow m})$. Hence, from~\eqref{eq:zero-comb-restricted}, we have 
\begin{align} \label{eq:restricted-equal}
\pi(T^{\rightarrow m}) = \pi'((T')^{\rightarrow m}).
\end{align}
Consider the pair of partitions $S_1 \sqcup \cdots S_{t - 1} \sqcup T \sqcup (S_t \setminus T) = [n]$ and $S'_1 \sqcup \cdots S'_{t - 1} \sqcup T' \sqcup (S'_t \setminus T') = [n]$. From the definition of $T$, we must have $T \ne \emptyset$ because $i^*$ belongs to $T$, and $(S_t \setminus T) \ne \emptyset$ becase $i_t$ does not belong to $T$. From this and~\eqref{eq:restricted-equal}, these partitions are non-empty and they match. However, these matching partitions have $t + 1$ parts, which contradicts the maximality of the number of parts of $S_1 \sqcup \cdots \sqcup S_t$ and $S'_1 \sqcup \cdots \sqcup S'_t$. This concludes our proof.
\end{proof}

\subsubsection{Proof of Theorem~\ref{thm:corank-bound}}

With the characterization from the previous subsection ready, we can now easily prove our main theorem of this section (Theorem~\ref{thm:corank-bound}). We will also use two simple inequalities regarding the multinomial coefficients stated below. For completeness, we provide their proofs in the appendix.

\begin{fact} \label{fact:multichoose-additive}
For every $a_1, \dots, a_k, a'_1, \dots, a'_k \in \N$, we have
\begin{align*}
\binom{a_1 + \cdots + a_k + a'_1 + \cdots + a'_k}{a_1 + a'_1, \dots, a_k + a'_k} \geq \binom{a_1 + \cdots + a_k}{a_1, \dots, a_k} \cdot \binom{a'_1 + \cdots + a'_k}{a'_1, \dots, a'_k}
\end{align*}
\end{fact}

\begin{fact} \label{fact:multichoose-ineq}
For every $k \in \N$ and $a_1, \dots, a_k \in \N$, we have
\begin{align*}
\binom{a_1 + \cdots + a_k}{a_1, \dots, a_k} \geq \left(\frac{a_1 + \cdots + a_k}{2}\right)^{\lfloor k/2 \rfloor}
\end{align*}
\end{fact}

\begin{proof}[Proof of Theorem~\ref{thm:corank-bound}]
Let us fix a pair of non-empty partitions $S_1 \sqcup \cdots \sqcup S_k = [n]$ and $S'_1 \sqcup \cdots \sqcup S'_k = [n]$ such that\footnote{We may assume that $|S_i| = |S'_i|$; otherwise, $\pi(S_i^{\to m})$ and $\pi'((S'_i)^{\to m})$ are obviously not equal and hence $S_1 \sqcup \cdots \sqcup S_k = [n]$ and $S'_1 \sqcup \cdots \sqcup S'_k = [n]$ do not match.} $|S_i| = |S'_i|$ for all $i \in [k]$. Notice that, when we pick $\pi: [mn] \to [mn]$ uniformly at random, $\left(\pi\left(S_1^{\to m}\right), \cdots, \pi\left(S_k^{\to m}\right)\right)$ is simply a random partition of $[mn]$ into subsets of size $m|S_1|, \dots, m|S_k|$. Hence, the probability that these partitions match is equal to
\begin{align*}
\frac{1}{\binom{mn}{m|S_1|, \dots, m|S_k|}}.
\end{align*}
Hence, by evoking Lemma~\ref{lem:corank-condition} and taking union bound over all pairs of partitions $S_1 \sqcup \cdots \sqcup S_k = [n]$ and $S'_1 \sqcup \cdots \sqcup S'_k = [n]$, we have
\begin{align*}
\Pr_{\pi, \pi' \sim \Pi_{mn}}[\defc(\bA_{\pi, \pi'}^T) \geq k] 
&\leq \sum_{\substack{S_1 \sqcup \cdots \sqcup S_k = [n], S'_1 \sqcup \cdots \sqcup S'_k = [n] \\ |S_1| = |S'_1| > 0, \dots, |S_k| = |S'_k| > 0}} \frac{1}{\binom{mn}{m|S_1|, \dots, m|S_k|}} \\
&= \sum_{\substack{a_1, \dots, a_k \in \N \\ a_1 + \cdots + a_k = n}} \sum_{\substack{S_1 \sqcup \cdots \sqcup S_k = [n], S'_1 \sqcup \cdots \sqcup S'_k = [n] \\ |S_1| = |S'_1| = a_1, \dots, |S_k| = |S'_k| = a_k}} \frac{1}{\binom{mn}{m a_1, \dots, m a_k}} \\
&= \sum_{\substack{a_1, \dots, a_k \in \N \\ a_1 + \cdots + a_k = n}} \frac{\binom{n}{a_1, \dots, a_k}^2}{\binom{mn}{ma_1, \dots, ma_k}} \\
\text{(Fact~\ref{fact:multichoose-additive})} &\leq \sum_{\substack{a_1, \dots, a_k \in \N \\ a_1 + \cdots + a_k = n}} \frac{1}{\binom{n}{a_1, \dots, a_k}^{(m - 2)}} \\
(\text{Fact}~\ref{fact:multichoose-ineq}) &\leq \sum_{\substack{a_1, \dots, a_k \in \N \\ a_1 + \cdots + a_k = n}} \frac{1}{\left(n / 2\right)^{(m - 2) \cdot \lfloor k/2 \rfloor}} \\
&\leq \frac{n^{k - 1}}{\left(n / 2\right)^{(m - 2) \cdot \lfloor k/2 \rfloor}} \\
&\leq \left(\frac{n^{2}}{(n/2)^{m - 2}}\right)^{\frac{k - 1}{2}} \qedhere
\end{align*}
\end{proof}

\section{Lower Bound Proofs}\label{sec:lb_pf}

In this section, we prove our lower bound on the number of messages (Theorem~\ref{th:message_lb}), which is a direct consequence of the following two theorems:
\begin{theorem} \label{thm:lognq-lowerbound}
     Suppose $\sigma \geq 1$. Then, for any $\sigma$-secure $n$-party aggregation protocol over $\inputring$ in which each party sends $m$ messages, we have $m = \Omega(\log_n q)$.
\end{theorem}

\begin{theorem} \label{thm:sigma-lowerbound}
For any $\sigma$-secure $n$-party aggregation protocol over $\inputring$ in which each party sends $m$ messages, we have $m = \Omega\left(\frac{\sigma}{\log(\sigma n)}\right)$.
\end{theorem}
We prove Theorem~\ref{thm:lognq-lowerbound} in Section~\ref{subsec:field}, while we prove Theorem~\ref{thm:sigma-lowerbound} in Section~\ref{subsec:sec_dep_lb}. Before we proceed to the proofs, let us start by proving the following fact that will be used in both proofs: the output of the encoder on a party's input must uniquely determine the input held by the party.
\begin{lemma} \label{lem:zerr}
 For any $n$-party aggregation protocol $\cP$ with encoder $\enc: \inputring \to [\ell]^m$, we have that for any $x, x' \in \inputring$ with $x\neq x'$, the distributions $\enclocal{x}{\enc}$ and $\enclocal{x'}{\enc}$ have disjoint supports. 
 
 As a consequence, for any output vector $\yy \in [\ell]^{nm}$, there exists at most one $\bx = (x_1, x_2, \dots, x_n) \in\inputring^n$ such that $\yy$ is a possible output $(\enc_{x_1}, \enc_{x_2}, \dots, \enc_{x_n})$.
\end{lemma}
\begin{proof}
For the sake of contradiction, suppose there exist $x,x'\in\inputring$ with $x\neq x'$ such that $\enclocal{x}{\enc}$ and $\enclocal{x'}{\enc}$ have a common element in the support, say $\mathbf{z}$. Then, let $\mathbf{z}' \in [\ell]^m$ be an element in the support of $\enclocal{0}{\enc}$. Note that it follows that $(\mathbf{z}, \underbrace{\mathbf{z}', \mathbf{z}', \dots, \mathbf{z}'}_{n-1})$ is a possible output of inputs $(x, \bzero^{n-1})$ and $(x', \bzero^{n-1})$, which means that the analyzer cannot uniquely determine the parties' inputs from the output, thereby contradicting the correctness of the protocol. This completes the proof.
\end{proof}

\subsection{Field-Dependent Bound} \label{subsec:field}
We now present the proof of Theorem~\ref{thm:lognq-lowerbound}. Recall from Section~\ref{subsec:overview} that $\sumset_s$ is defined as $\{\bx \in \F_q^n \mid \sum_i x_i = s\}$. The key technical lemma is the following.
\begin{lemma}\label{lem:nqtech}
    For each $s\in \inputring$ and every $n$-user one-round aggregation protocol $\cP$ in the anonymized model with encoder $\enc: \inputring \to [\ell]^m$, there exists a pair of inputs $\bb, \bb' \in \sumset_s$ such that $\sd\left(\encshufdist{\bb}{\enc}, \encshufdist{\bb'}{\enc}\right) \geq 1 - n^{nm}/q^{n-1}$.
\end{lemma}

Throughout this subsection, let us fix $s \in \F_q$.
Before proving Lemma~\ref{lem:nqtech}, we first define some notation. For every possible shuffler output vector $\yy$ and input $\bb\in \sumset_s$, let $p_{\bb,\yy}$ denote the probability that on input $\bb$ the encoder outputs $\yy$, i.e., $\Pr_{Y \sim \encshufdist{\bb}{\enc}}[Y = \by]$. Moreover, let $\invset_{\yy} = \{ \bb\in \sumset_s \; | \; p_{\bb,\yy} > 0\}$ denote the set of sum-$s$ inputs that are possible given that the output is $\yy$.

\begin{lemma}\label{lemma:feasible_inputs}
    $|\invset_{\yy}| \leq n^{nm}$.
\end{lemma}
\begin{proof}
    Suppose $\yy$ is an output vector consisting of $nm$ messages with $|\invset_{\yy}|>0$. Consider a function $g: [nm] \rightarrow [n]$ that associates each of the $mn$ messages to a single party. Note that $\yy$ and $g$ uniquely identify the set of messages $Y_i$ sent by each party $i$.
    In turn, $Y_i$ must correspond to a unique input $x_i$ to party $i$ by Lemma~\ref{lem:zerr}. Then, it follows that $\yy$ and $g$ can determine at most one input $\bb\in \invset_{\yy}$.
    Since there are at most $n^{nm}$ valid functions $g$, the desired bound on $|\invset_{\yy}|$ follows.
\end{proof}

Let $p_{\yy} = \sum_{\bb\in\sumset_s} p_{\bb,\yy}$, and define $d_{\yy} = \tfrac{1}{q^{2n - 2}} \sum_{\bb\in \sumset_s}\sum_{\bb'\in \sumset_s} |p_{\bb,\yy} - p_{\bb',\yy}|$ as the average difference between probabilities $p_{\bb,\yy}$ and $p_{\bb',\yy}$ over all pairs of inputs $\bb, \bb'$ with sum $s$. Then, we have the following lemma.
\begin{lemma}\label{lemma:column_distance}
    %If $N^{n-1} > n^{nm}$ then 
    $d_{\yy} \geq 2\left(1-\frac{n^{nm}}{q^{n - 1}}\right) p_{\yy} / q^{n - 1}$.
\end{lemma}
\begin{proof} We have
    \begin{align*}
    q^{2n - 2} d_{\yy} &\geq 2 \sum_{\bb\in \invset_{\yy}}\sum_{\bb'\in \sumset_s\setminus \invset_{\yy}} |p_{\bb,\yy}- \bf{0}|\\
    &= 2\, |\sumset_s \setminus \invset_{\yy}| \sum_{\bb\in \invset_{\yy}} p_{\bb,\yy}\\
    &= 2 \left(q^{n-1} - |\invset_{\yy}|\right) p_{\yy}\\
    \text{(Lemma~\ref{lemma:feasible_inputs})} &\geq 2 \left(q^{n-1} - n^{nm}\right) p_\yy. \qedhere
    \end{align*}
\end{proof}
We now prove Lemma~\ref{lem:nqtech}.
\begin{proof}[Proof of Lemma~\ref{lem:nqtech}]
We will in fact show the stronger statement that the (scaled) \emph{average} statistical distance for pairs of inputs in $\sumset_s$ is lower bounded by $1 - n^{nm}/q^{n-1}$, i.e.,
\[
  \davg \geq 1 - \frac{n^{nm}}{q^{n-1}},
\]
where
\begin{equation}
  \davg = \frac{1}{q^{2n-2}} \sum_{\bb\in\sumset_s} \sum_{\bb'\in\sumset_s} \sd\left(\encshufdist{\bb}{\enc}, \encshufdist{\bb'}{\enc}\right). \label{eq:davg}
\end{equation}
Note that by Lemma~\ref{lemma:column_distance}, we have
\begin{align*}
    \davg &= \sum_{\yy} \frac{d_\yy}{2}\\
          &\geq \frac{1}{q^{n-1}} \left(1-\frac{n^{nm}}{q^{n - 1}}\right) \sum_{\yy} p_{\yy}\\
          &\geq \frac{1}{q^{n-1}} \left(1-\frac{n^{nm}}{q^{n - 1}}\right) \sum_{\yy} \sum_{\bb\in\sumset_s} p_{\bb,\yy}\\
          &= 1-\frac{n^{nm}}{q^{n - 1}},
\end{align*}
where the last line follows from the fact that $|\sumset_s| = q^{n-1}$. To conclude, note that it follows that at least one of the summands in \eqref{eq:davg} must be at least $1-\frac{n^{nm}}{q^{n - 1}}$, as desired.
\end{proof}

Theorem~\ref{thm:lognq-lowerbound} now follows easily from Lemma~\ref{lem:nqtech}.
\begin{proof}[Proof of Theorem~\ref{thm:lognq-lowerbound}]
 Suppose $\cP$ is such a $\sigma$-secure $n$-party aggregation protocol with encoder $\enc: \F_q \to [\ell]^m$. Then, choose an arbitrary $s\in\inputring$. Note that by Lemma~\ref{lem:nqtech}, there exist $\bb, \bb' \in \sumset_s$ such that
 \[
   2^{-\sigma} \geq \sd\left(\encshufdist{\bb}{\enc}, \encshufdist{\bb'}{\enc}\right) \geq 1 - \frac{n^{nm}}{q^{n-1}}.
 \]
Thus, if $\sigma \geq 1$, it follows that $m = \Omega(\log_n q)$, as desired.
\end{proof}

\subsection{Security-Dependent Bound}\label{subsec:sec_dep_lb}
We now turn to the proof of Theorem~\ref{thm:sigma-lowerbound}, which follows from the next theorem.
\begin{theorem} \label{thm:lower-bound-n-to-m}
Let $\enc$ be the encoder of any summation protocol for $n > 2$ parties with $m$ messages sent per party. Then, there exists a vector $\bx \in \sumset_0$ such that the statistical distance between $\encshufdist{\bzero}{\enc}$ and $\encshufdist{\bx}{\enc}$ is at least $\frac{1}{(10nm)^{5m}}$.
\end{theorem}
It is not hard to see that Theorem~\ref{thm:sigma-lowerbound} follows from Theorem~\ref{thm:lower-bound-n-to-m}:
\begin{proof}[Proof of Theorem~\ref{thm:sigma-lowerbound}]
 Simply note that by Theorem~\ref{thm:lower-bound-n-to-m} and the definition of $\sigma$-security, we can find $\bx\in\sumset_0$ such that
 \[
   2^{-\sigma} \geq \sd\left(\encshufdist{\bzero}{\enc}, \encshufdist{\bx}{\enc}\right) \geq \frac{1}{(10nm)^{5m}},
 \]
 which immediately implies that $m = \Omega\left(\frac{\sigma}{\log(\sigma n)}\right)$, as desired.
\end{proof}
Henceforth, we focus on proving Theorem~\ref{thm:lower-bound-n-to-m}.

\subsubsection*{Warm-up: Proof of Theorem~\ref{thm:lower-bound-n-to-m} for Ishai et al.'s protocol.}
Before we prove Theorem~\ref{thm:lower-bound-n-to-m} for the general case, let us sketch a proof specific to Ishai et al.'s protocol. The input vector $\bx$ we will use is simply $\bx = (1, \cdots, 1, -(n - 1))$.

To lower bound $\sd(\encshufdist{\bzero}{}, \encshufdist{\bx}{})$, we give a ``distinguisher'' $\cA$ that takes in the output $(y_1, \dots, y_{\pi(mn)})$ of the shuffler and outputs either 1 (i.e. ``accept'') or 0 (i.e. ``reject''). Its key property will be that the probability that $\cA$ accepts when $(y_{\pi(1)}, \dots, y_{\pi(mn)}) \sim \encshufdist{\bzero}{}$ is more than that of when $(y_{\pi(1)}, \dots, y_{\pi(mn)}) \sim \encshufdist{\bx}{}$ by an additive factor of $\frac{1}{(en)^m}$. This immediately implies that the distributions $\encshufdist{\bzero}{}$ and $\encshufdist{\bx}{}$ are at a statistical distance of at least $\frac{1}{(en)^m}$ as well. (Note that this bound is slightly better than the one in Theorem~\ref{thm:lower-bound-n-to-m}.)

The distinguisher $\cA$ is incredibly simple here: $\cA$ accepts iff $y_{\pi(1)} + \cdots + y_{\pi(m)} = 0$. To see that it satisfies the claim property, observe that, when $\pi(1), \dots, \pi(m)$ not all come from the same party, $y_{\pi(1)} + \cdots + y_{\pi(m)}$ is simply a random number in $\F_q$, meaning that $\cA$ accepts with probability $1/q$ (in both distributions). On the other hand, when $\pi(1), \dots, \pi(m)$ come from the same party, $y_{\pi(1)} + \cdots + y_{\pi(m)}$ is always zero in the distribution $\encshufdist{\bzero}{}$ and hence $\cA$ always accept. For the distribution $\encshufdist{\bx}{}$, if $\pi(1), \dots, \pi(m)$ comes from the same party $i \ne n$, then the sum $y_{\pi(1)} + \cdots + y_{\pi(m)}$ is always one and hence $\cA$ rejects. Thus, the probability that $\cA$ accepts in the former distribution is more than that of the latter by an additive factor of $\frac{n - 1}{\binom{nm}{m}} \geq \frac{1}{(en)^m}$. (The -1 factor corresponds to the case where $p(1), \cdots, p(m)$ comes from party $i = n$; here $\cA$ might accept if $- (n - 1) = 0$ in $\F_q$.) This concludes the proof sketch.

\subsubsection*{From Ishai et al.'s protocol to general protocols.}
Having sketched the argument for Ishai et al.'s protocol, one might wonder whether the same approach would work for general protocols. In particular, here instead of checking if $y_{\pi(1)} + \cdots + y_{\pi(m)} = 0$, we would check whether $y_{\pi(1)}, \dots, y_{\pi(m)}$ is a valid output of the encoder when the input is zero. Now, the statement for when $\pi(1), \dots, \pi(m)$ comes from the same party remains true. However, the issue is that, when $\pi(1), \dots, \pi(m)$ do not all come from the same party, it is not necessarily true that the acceptance probability of $\cA$ would be the same for both distributions.

To avoid having these ``cross terms'' affect the probability of acceptance of $\cA$ too much, we pick the smallest integer $t$ such that the ``$t$-message marginals'' (defined formally below) of $\enclocal{0}{\enc}$ and $\enclocal{1}{\enc}$ differ ``substantially''. Then, we modify $\cA$ so that it performs an analogous check on $y_{\pi(1)}, \dots, y_{\pi(t)}$ (instead of $y_{\pi(1)}, \dots, y_{\pi(m)}$ as before). Once again, we will have that, if $\pi(1), \dots, \pi(t)$ corresponds to the same party, then the probability that $\cA$ accepts differs significantly between the two cases. On the other hand, due to the minimality of $t$, we can also argue that, when $\pi(1), \dots, \pi(t)$ are not all from the same parties (i.e. ``cross terms''), the difference is small. Hence, the former case would dominate and we can get a lower bound on the difference as desired. This is roughly the approach we take in the proof of Theorem~\ref{thm:lower-bound-n-to-m} below. There are subtle points we have to change in the actual proof below. For instance, we cannot simply use the input $(1, \cdots, 1, -(n - 1))$ as in the case of Ishai et al. protocol because, if the $t$-marginal of $\enclocal{-(n - 1)}{\enc}$ deviates from $\enclocal{0}{\enc}$ more substantially than that of $\enclocal{1}{\enc}$, then this could affect the acceptance probability by a lot. Hence, in the actual proof, we instead pick $x^*$ that minimizes the value of such $t$ among all numbers in $\F_q$, and use the input vector $\bx = (x^*, \dots, x^*,- (n - 1)x^*)$.

\subsubsection*{Additional Notation and Observation.} To formally prove Theorem~\ref{thm:lower-bound-n-to-m} in the general form, we need to formally define the notion of $t$-marginal. For a distribution $\cD$ supported on $[\ell]^m$ and a positive integer $t \leq m$, its \emph{$t$-marginal}, denoted by $\cD|_t$, supported on $[\ell]^t$ is simply the marginal of $\cD$ on the first $t$-coordinates; more formally, for all $\by \in [\ell]^t$, we have
\begin{align*}
\Pr_{Y \sim \cD|_t}[Y = \by] = \sum_{y_{t + 1}, \dots, y_m \in [\ell]} \Pr_{Y \sim \cD}[Y = \by \circ (y_{t + 1}, \dots, y_m)].
\end{align*}

An observation that will simplify our proof is that we may assume w.l.o.g. that the distribution $\enclocal{x}{\enc}$ for every $x \in \F_q$ is permutation invariant, i.e., that for any $\pi: [m] \to [m]$ and any $\by \in [\ell]^m$, we have
\begin{align*}
\Pr_{Y \sim \enclocal{x}{\enc}}[Y = \by] = \Pr_{Y \sim \enclocal{x}{\enc}}[Y = \pi(\by)].
\end{align*}
This is because we may apply a random permutation to the encoding $\enc_x$ before sending it to the shuffler, which does not change the distribution $\encshufdist{\enc}{x}$. Notice that our observation implies that $\enclocal{x}{\enc}|_t$ is also permutation invariant.

%Finally, for $\bx \in \F_q^n$, we write $\enclocal{\bx}{\enc}$ to denote the distribution on $[\ell]^{mn}$ where the coordinates $m(i - 1) + 1, \dots, mi$ are sampled i.i.d. from $\enclocal{x_i}{\enc}$ for all $i \in \{1, \dots, n\}$. (In other words, this is the distribution of the parties' outputs when the inputs are $x_1, \dots, x_n$ before being randomly permuted by the shuffler.) We also use $\encshufdist{\bx}{\enc}$ to denote the distribution on $[\ell]^{mn}$ which results from applying a random permutation on $\by$ sampled from $\enclocal{\bx}{\enc}$. Note that this is indeed the distribution of the view of the analyzer.

\begin{proof}[Proof of Theorem~\ref{thm:lower-bound-n-to-m}]
Let $t \leq m$ be the smallest positive integer such that \\ $\max_{x \in \F_q} \sd(\enclocal{0}{\enc}|_t, \enclocal{x}{\enc}|_{t})$ is at least $\frac{1}{(10nm)^{4(m - t)}}$. Note that such $t$ always exist because the requirement holds for $t = m$, at which $\enclocal{0}{\enc}|_t = \enclocal{0}{\enc}$ and $\enclocal{1}{\enc}|_t = \enclocal{1}{\enc}$ have statistical distance 1 (as their supports are disjoint due to Lemma~\ref{lem:zerr}).

For $t$ as defined above, let $x^* = \argmax_{x \in \F_q} \sd(\enclocal{0}{\enc}|t, \enclocal{x}{\enc}|_t)$ and let us defined $H$ as the set of elements of $[\ell]^t$ whose probability under $\enclocal{0}{\enc}|_t$ is higher than under $\enclocal{x^*}{\enc}|_t$. More formally, $H = \{\by \in [\ell]^t : \enclocal{0}{\enc}|_t(\by) > \enclocal{x^*}{\enc}|_t(\by)\}$. By definition of statistical distance, we have
\begin{align} \label{eq:tvd}
\Pr_{\by \in \enclocal{0}{\enc}|_t}[\by \in H] - \Pr_{\by \in \enclocal{x^*}{\enc}|_t}[\by \in H] = SD(\enclocal{0}{\enc}|_t, \enclocal{x^*}{\enc}|_t) \geq \frac{1}{(10nm)^{4(m - t)}},
\end{align}
where the inequality follows from our choice of $t$.

Let $\bx = (x^*, \dots, x^*, -(n - 1)x^*)$; clearly, $\bx \in \sumset_0$ as desired. We next give a distinguisher for the distributions $\encshufdist{\bzero}{\enc}$ and $\encshufdist{\bx}{\enc}$. The distinguisher $\cA$ takes in the permuted output $(y_{\pi(1)}, \dots, y_{\pi(nm)})$. It returns one (i.e., ``accept'') if $(y_{\pi(1)}, \dots, y_{\pi(t)})$ belongs to $H$ and it returns zero (i.e., ``reject'') otherwise.

We will show that the probability that $\cA$ accepts on $\encshufdist{\bzero}{\enc}$ is more than the probability that it accepts on $\encshufdist{\bx}{\enc}$ by at least $\frac{1}{(10nm)^{5m}}$, which implies that the statistical distance between $\encshufdist{\bzero}{\enc}$ and $\encshufdist{\bx}{\enc}$ is also at least $\frac{1}{(10nm)^{5m}}$ as desired.

To argue about the acceptance probability of $\cA$, it is worth noting that there are two sources of randomness here: the output $\by$ (sampled from $\enclocal{\bzero}{\enc}$ or $\enclocal{\bx}{\enc}$) and the permutation $\pi$. More formally, we may write the probability that $\cA$ accepts on $\encshufdist{\bzero}{\enc}$ and that on $\encshufdist{\bx}{\enc}$ as 
\begin{align*}
\Pr_{\pi \sim \Pi_{mn}, \by \sim \enclocal{\bzero}{\enc}}[\cA(\pi(y)) = 1].
\end{align*}
and 
\begin{align*}
\Pr_{\pi \sim \Pi_{mn}, \by \sim \enclocal{\bx}{\enc}}[\cA(\pi(y)) = 1].
\end{align*}
respectively.
Hence, the difference between the probability that $\cA$ accepts on $\encshufdist{\bzero}{\enc}$ and that on $\encshufdist{\bx}{\enc}$ is
\begin{align*}
&\Pr_{\pi \sim \Pi_{mn}, \by \sim \encshufdist{\bzero}{\enc}}[\cA(\pi(y)) = 1] - \Pr_{\pi \sim \Pi_{mn}, \by \sim \encshufdist{\bx}{\enc}}[\cA(\pi(y)) = 1] \\
&= \E_{\pi \sim \Pi_{mn}}\left[\Pr_{\by \sim \encshufdist{\bzero}{\enc}}[\cA(\pi(y)) = 1] - \Pr_{\by \sim \encshufdist{\bx}{\enc}}[\cA(\pi(y)) = 1]\right].
\end{align*}

For brevity, let us define $\Delta_{\pi}$ as
\begin{align*}
\Delta_{\pi} := \Pr_{\by \sim \encshufdist{\bzero}{\enc}}[\cA(\pi(y)) = 1] - \Pr_{\by \sim \encshufdist{\bx}{\enc}}[\cA(\pi(y)) = 1].
\end{align*}
Note that the quantity we would like to lower bound is now simply $\E_\pi[\Delta_\pi]$.

For each party $i \in \{1, \dots, n\}$ and any permutation $\pi: [mn] \to [mn]$, we use $U^i_{\pi}$ to denote $\{\pi(1), \dots, \pi(t)\} \cap \{m(i - 1) + 1, \dots , mi\}$.
Furthermore, we define \emph{the largest number of messages from a single party} for a permutation $\pi$ as $C_{\pi} := \max_{i = 1, \dots, n} |U^i_{\pi}|$.

In the next part of the proof, we classify $\pi$ into three categories, as listed below. For each category, we prove either a lower or an upper bound on $\Delta_\pi$ and the probability that a random permutation falls into that category.
\begin{enumerate}[I.]
\item $C_{\pi} = t$ and $|U_\pi^n| \ne t$. In other words, all of $\{\pi(1), \dots, \pi(t)\}$ correspond to a single party and that party is not the last party.
\item $C_{\pi} = t$ and $|U_\pi^n| = t$. In other words, all of $\{\pi(1), \dots, \pi(t)\}$ correspond to the last party $n$.
\item $C_{\pi} < t$. Not all of $\pi(1), \dots, \pi(t)$ comes from the same party.
\end{enumerate}
We will show that for category I permutations, $\Delta_\pi$ is large (Lemma~\ref{lem:dist-upper}) and the probability that a random permutation belongs to this category is not too small (Lemma~\ref{lem:prob-large-marginal}). For both categories II and III, we show that $|\Delta_\pi|$ is small (Lemmas~\ref{lem:dist-last-party} and~\ref{lem:dist-upper}) and the probabilities that a random permutation belongs to each of these two categories are not too large (Lemmas~\ref{lem:prob-last-party} and~\ref{lem:prob-bound-small-marginal}).

These quantitative bounds are such that the first category dominates $\E_\pi[\Delta_\pi]$, meaning that we get a lower bound on this expectation as desired; this is done at the very end of the proof.

\paragraph{Category I: } $C_{\pi} = t$ and $|U_\pi^n| \ne t$.

We now consider the first case: when $\{\pi(1), \dots, \pi(t)\}$ corresponds to a single party $i \ne n$. In this case, $\Delta_\pi$ is exactly equal to the statistical distance between $\enclocal{0}{\enc}$ and $\enclocal{x^*}{\enc}$ (which we know from~\eqref{eq:tvd} to be large):

\begin{lemma} \label{lem:dist-lower}
For any $\pi$ such that $C_{\pi} = t$ and $|U^n_\pi| \ne t$, we have
\begin{align*}
\Delta_{\pi} = \sd(\enclocal{0}{\enc}|_t, \enclocal{x^*}{\enc}|_t).
\end{align*}
\end{lemma}

\begin{proof}
Let $i \in \{1, \dots, n\}$ be the party such that $|U^i_\pi| = C_{\pi} = t$. When $\by$ is drawn from $\encdist{\bx}{\enc}$ (respectively $\encdist{\bzero}{\enc}$), $\{\pi(1), \cdots, \pi(t)\} \subseteq \{m(i - 1) + 1, \dots, mi\}$, it is the case that $(y_{\pi(1)}, \dots, y_{\pi(t)})$ is simply distributed as $\encdist{x_i}{\enc}|_t$ (respectively $\encdist{0}{\enc}|_t$). Recall that we assume that $U^n_{\pi} \ne t$, which means that $i \ne n$ or equivalently $x_i = x^*$. Hence, we have 
\begin{align*}
\Pr_{\by \sim \encdist{\bx}{\enc}}[\cA(\pi(\by)) = 1] = \Pr_{\by' \sim \enclocal{x_i}{\enc}|_t}[\by' \in H] = \Pr_{\by' \sim \enclocal{x^*}{\enc}|_t}[\by' \in H].
\end{align*}
and 
\begin{align*}
\Pr_{\by \sim \encdist{\bzero}{\enc}}[\cA(\pi(\by)) = 1] = \Pr_{\by' \sim \enclocal{0}{\enc}|_t}[\by' \in H].
\end{align*}
Combining the above two equalities with~\eqref{eq:tvd} implies that $\Delta_{\pi} = \sd(\enclocal{0}{\enc}|_t, \enclocal{x^*}{\enc}|_t)$ as desired.
\end{proof}

The probability that $\pi$ falls into this category can be simply computed:

\begin{lemma} \label{lem:prob-large-marginal}
$\Pr_{\pi}[C_{\pi} = t \wedge U_\pi^n \ne t] = \frac{(n - 1) \cdot \binom{m}{t}}{\binom{nm}{t}}$.
\end{lemma}

\begin{proof}
$C_{\pi} = t$ and $|U^n_\pi| \ne t$ if and only if there exists a party $i \in \{1, \dots, n - 1\}$ such that $\pi(\{1, \dots, t\}) \subseteq \{m(i - 1) + 1, \dots, mi\}$. For a fixed $i$, this happens with probability $\frac{\binom{m}{t}}{\binom{nm}{t}}$. Notice also that the event is disjoint for different $i$'s. As a result, the total probability that this event occurs for at least one $i$ is $(n - 1) \cdot \frac{\binom{m}{t}}{\binom{nm}{t}}$.
\end{proof}

\paragraph{Category II: } $C_{\pi} = t$ and $|U_\pi^n| = t$.

We now consider the second category: when $\{\pi(1), \dots, \pi(t)\}$ corresponds to the last party $n$. In this case, our choice of $x^*$ implies that $|\Delta_\pi|$ is upper bounded by the statistical distance between $\enclocal{0}{\enc}|_t$ and $\enclocal{x^*}{\enc}|_t$, as formalized below.

\begin{lemma} \label{lem:dist-last-party}
For any $\pi$ such that $C_{\pi} = t$ and $|U^n_\pi| = t$, we have
\begin{align*}
|\Delta_{\pi}| \leq \sd(\enclocal{0}{\enc}|_t, \enclocal{x^*}{\enc}|_t).
\end{align*}
\end{lemma}

\begin{proof}
In this case, we have $\{\pi(1), \cdots, \pi(i)\} \subseteq \{m(n - 1) + 1, \dots, mn\}$. Thus, when $\by$ is drawn from $\encdist{\bx}{\enc}$ (respectively $\encdist{\bzero}{\enc}$), it is the case that $(y_{\pi(1)}, \dots, y_{\pi(t)})$ is simply distributed as $\enclocal{x_n}{\enc}|_t$ (respectively $\enclocal{0}{\enc}|_t$). Hence, we have 
\begin{align*}
\Pr_{\by \sim \encdist{\bx}{\enc}}[\cA(\pi(\by)) = 1] = \Pr_{\by' \sim \enclocal{x_n}{\enc}|_t}[\by' \in H]
\end{align*}
and 
\begin{align*}
\Pr_{\by \sim \encdist{\bzero}{\enc}}[\cA(\pi(\by)) = 1] = \Pr_{\by' \sim \enclocal{0}{\enc}|_t}[\by' \in H].
\end{align*}
Combining the above two equalities implies that $|\Delta_{\pi}| \leq \sd(\enclocal{0}{\enc}|_t, \enclocal{x_n}{\enc}|_t)$. Recall that $x^*$ is chosen to maximize $\sd(\enclocal{0}{\enc}|_t, \enclocal{x^*}{\enc}|_t)$, which means that $\sd(\enclocal{0}{\enc}|_t, \enclocal{x_n}{\enc}|_t)$ $\leq \sd(\enclocal{0}{\enc}|_t, \enclocal{x^*}{\enc}|_t)$. Hence, we have $|\Delta_{\pi}| \leq \sd(\enclocal{0}{\enc}|_t, \enclocal{x^*}{\enc}|_t)$ as desired.
\end{proof}

The probability that $\pi$ falls into this category can be simply computed in a similar manner as in the first case:

\begin{lemma} \label{lem:prob-last-party}
$\Pr_{\pi}[C_{\pi} = t \wedge |U_\pi^n| = t] = \frac{\binom{m}{t}}{\binom{nm}{t}}$.
\end{lemma}

\begin{proof}
$C_{\pi} = t$ and $|U^n_\pi| = t$ if and only if $\pi(\{1, \dots, t\}) \subseteq \{m(n - 1) + 1, \dots, mn\}$. This happens with probability exactly $\frac{\binom{m}{t}}{\binom{nm}{t}}$.
\end{proof}

\paragraph{Category III: } $C_\pi < t$.

Finally, we consider any permutation $\pi$ such that not all of $\{\pi(1), \dots, \pi(t)\}$ correspond to a single party. On this front, we may use our choice of $t$ to give an upper bound on $|\Delta_\pi|$ as follows.

\begin{lemma} \label{lem:dist-upper}
For any $\pi$ such that $C_{\pi} < t$, we have
\begin{align*}
|\Delta_{\pi}| < m \cdot \frac{1}{(10nm)^{4(m - C_{\pi})}}.
\end{align*}
\end{lemma}

\begin{proof}
In fact, we will show something even stronger: that the statistical distance of $(y_{\pi(1)}, \dots, y_{\pi(t)})$ when $\by$ is drawn from $\encdist{\bzero}{\enc}$ and that when $\by$ is drawn from $\encdist{\bx}{\enc}$ is at most $m \cdot \frac{1}{(10nm)^{4(m - C_{\pi})}}$. The desired bound immediately follows.

Let $I$ denote the set of all parties $i$ such that $U_i \ne \emptyset$. Observe that, when $\by$ is drawn from $\encdist{\bx}{\enc}$ (respectively $\encdist{\bzero}{\enc}$), $(y_p)_{p \in U_i}$ is simply distributed as $\enclocal{x_i}{\enc}|_{|U_i|}$ (respectively $\enclocal{0}{\enc}|_{|U_i|}$) and that these are independent for different $i$. In other words, $(y_{\pi(1)}, \dots, y_{\pi(t)})$ is (after appropriate rearrangement) just the product distribution $\prod_{i \in I} \enclocal{x_i}{\enc}|_{|U_i|}$ (respectively $\prod_{i \in I} \enclocal{0}{\enc}|_{|U_i|}$).

Recall from the definition of $C_{\pi}$ that $|U_i|$ is at most $C_{\pi}$ for all $i$. Since $C_{\pi} < t$ and from our choice of $t$, we must have $\sd(\enclocal{0}{\enc}|_{|U_i|}, \enclocal{x_i}{\enc}|_{|U_i|}) < \frac{1}{(10nm)^{4(m - C_{\pi})}}$ for all $i \in I$. Hence, we also have
\begin{align*}
SD\left(\prod_{i \in I} \enclocal{0}{\enc}|_{|U_i|}, \prod_{i \in I} \enclocal{x_i}{\enc}|_{|U_i|}\right) < |I| \cdot \frac{1}{(10nm)^{4(m - C_{\pi})}} \leq m \cdot \frac{1}{(10nm)^{4(m - C_{\pi})}},
\end{align*}
which concludes the proof.
\end{proof}

Next, we bound the probability that a random permutation $\pi$ belong to this category:

\begin{lemma} \label{lem:prob-bound-small-marginal}
For all $j < t$, we have
$\Pr_{\pi}[C_{\pi} = j] \leq \frac{n \cdot \binom{m}{t}}{\binom{nm}{t}} \cdot (nm)^{3(t - j)}$.
\end{lemma}

\begin{proof}
If $C_{\pi} = j$, there must exist a subset $T \subseteq \{1, \dots, t\}$ of size $j$ and a party $i \in \{1, \dots, n\}$ such that $\pi(T) \subseteq \{m(i - 1) + 1, \dots, mi\}$. For a fixed $T$ and $i$, this happens with probability exactly $\frac{\binom{m}{j}}{\binom{nm}{j}}$. Hence, by union bound over all $T$ and $i$, we have
\begin{align*}
\Pr_{\pi}[C_{\pi} = j] &\leq n \cdot \binom{t}{j} \cdot \frac{\binom{m}{j}}{\binom{nm}{j}} \\
&\leq \frac{n \cdot \binom{m}{t}}{\binom{nm}{t}} \cdot \frac{\binom{t}{j} \cdot m^{t - j}}{(nm)^{j - t}} \\
&\leq \frac{n \cdot \binom{m}{t}}{\binom{nm}{t}} \cdot (nm)^{3(t - j)}. \qedhere
\end{align*}
\end{proof}
\paragraph{Putting things together.} With all the claims ready, it is now simple to finish the proof of Theorem~\ref{thm:lower-bound-n-to-m}.
The difference between the probability that $\cA$ accepts on $\encshufdist{\bzero}{\enc}$ and that on $\encshufdist{\bx}{\enc}$ is
\begin{align*}
\E_{\pi}[\Delta_{\pi}] &= \Pr_{\pi}[C_{\pi} = t \wedge |U_\pi^n| \ne t] \cdot \E_{\pi}[\Delta_{\pi} \mid C_{\pi} = t \wedge |U_\pi^n| \ne t] \\
&\qquad + \Pr_{\pi}[C_{\pi} = t \wedge |U_\pi^n| = t] \cdot \E_{\pi}[\Delta_{\pi} \mid C_{\pi} = t \wedge |U_\pi^n| = t] \\ 
&\qquad + \sum_{j=1}^{t-1} \Pr_{\pi}[C_{\pi} = j] \cdot \E_{\pi}[\Delta_{\pi} \mid C_{\pi} = j] \\
(\text{Lemmas}~\ref{lem:dist-lower},\ref{lem:prob-large-marginal},\ref{lem:dist-last-party},\ref{lem:prob-last-party}) &\geq \frac{(n - 1) \cdot \binom{m}{t}}{\binom{nm}{t}} \cdot \sd(\enclocal{0}{\enc}|_t, \enclocal{x^*}{\enc}|_t) - \frac{\binom{m}{t}}{\binom{nm}{t}} \cdot \sd(\enclocal{0}{\enc}|_t, \enclocal{x^*}{\enc}|_t)  \\
&\qquad + \sum_{j=1}^{t-1} \Pr_{\pi}[C_{\pi} = j] \cdot \E_{\pi}[\Delta_{\pi} \mid C_{\pi} = j] \\
(\text{From}~n \geq 3) &\geq \frac{n \cdot \binom{m}{t}}{3 \binom{nm}{t}} \cdot \sd(\enclocal{0}{\enc}|_t, \enclocal{x^*}{\enc}|_t) + \sum_{j=1}^{t-1} \Pr_{\pi}[C_{\pi} = j] \cdot \E_{\pi}[\Delta_{\pi} \mid C_{\pi} = j] \\
(\eqref{eq:tvd}~\text{and Lemma}~\ref{lem:dist-upper}) &\geq \frac{n \cdot \binom{m}{t}}{3 \binom{nm}{t}} \cdot \frac{1}{(10nm)^{4(m - t)}} - \sum_{j=1}^{t-1} \frac{\Pr_{\pi}[C_{\pi} = j] \cdot m}{(10nm)^{4(m - j)}} \\
(\text{Lemma}~\ref{lem:prob-bound-small-marginal}) &\geq \frac{n \cdot \binom{m}{t}}{\binom{nm}{t}} \cdot \left(\frac{1}{3} \cdot \frac{1}{(10nm)^{4(m - t)}} - \sum_{j=1}^{t-1} \frac{(nm)^{3(t - j)} m}{(10nm)^{4(m - j)}}\right) \\
&\geq \frac{n \cdot \binom{m}{t}}{\binom{nm}{t}} \cdot \left(\frac{1}{3} - \sum_{j=1}^{t-1} \frac{1}{10^{t - j}}\right) \cdot \frac{1}{(10nm)^{4(m - t)}} \\
&\geq \frac{n \cdot \binom{m}{t}}{\binom{nm}{t}} \cdot \frac{1}{10} \cdot \frac{1}{(10nm)^{4(m - t)}} \\
&\geq \frac{1}{(nm)^t} \cdot \frac{1}{10} \cdot \frac{1}{(10nm)^{4(m - t)}} \\
&\geq \frac{1}{(10nm)^{5m}}. \qedhere
\end{align*}
\end{proof}

\section{Conclusion and Open Questions}
\label{sec:conclusion}

In this work, we provide an improved analysis for the split and mix protocol of Ishai et al.~\cite{balle_privacy_2019} in the shuffled model. Our analysis reduces the number of messages required by the protocol by a logarithmic factor. Moreover, for a large range of parameters, we give an asymptotically tight lower bound in terms of the number of messages that each party needs to send for \emph{any} protocol for secure summation. 

Although our lower bound is tight in terms of the number of messages, it does not immediately imply any communication lower bound beyond the trivial $\log q$ bound. For instance, when $q = n^{\log n}$ and $\sigma$ is a constant, then the number of messages needed by Ishai et al.'s protocol is $O\left(\frac{\log q}{\log n}\right) = O(\log n)$ but each message is also of length $O(\log q)$. However, our lower bound does not preclude a protocol with the same number of messages but of length only $O(\log n)$ bits. It remains an interesting open question to close this gap.

Another interesting open question is whether we can give a lower bound for $(\varepsilon, \delta)$-differentially private summation protocols when $\varepsilon$ is a constant. Currently, our lower bound does not give anything in this regime. In fact, to the best of our knowledge, it remains possible that an $(\varepsilon, 0)$-differentially private summation protocol exists where each party sends only $O_{\varepsilon}(\log n)$ bits. Coming up with such a protocol, or proving that one does not exists, would be a significant step in understanding the power of differential private algorithms in the shuffled model.

\bibliographystyle{splncs04}
\bibliography{constant_multi_message_summation}




% LNCS Homepage, \url{http://www.springer.com/lncs}. Last accessed 4
% \end{thebibliography}

\appendix

\section{Proofs of Bounds for Multinomial Coefficients}\label{sec:multinomial_bds_pfs}

Below we prove Facts~\ref{fact:multichoose-additive} and~\ref{fact:multichoose-ineq} from Section~\ref{sec:ub_pf}.

\begin{proof}[Proof of Fact~\ref{fact:multichoose-additive}]
Let $U = [a_1 + a'_1 + \cdots + a_k + a'_k], A = [a_1 + \cdots + a_k]$ and $B = U \setminus A$.

Consider the following process of generating a partition $S_1 \sqcup \cdots \sqcup S_k = U$. First, take a partition $T_1 \sqcup \cdots \sqcup T_k = A$ and a partition $T'_1 \sqcup \cdots \sqcup T'_k = B$. Then, let $S_i = T_i \cup T'_i$ for all $i \in [k]$.

Notice that each pair of $T_1 \sqcup \cdots \sqcup T_k$ with $|T_i| = a_i$ and $T'_1 \sqcup \cdots \sqcup T'_k$ with $|P_i| = a'_i$ produces different $S_1 \sqcup \cdots \sqcup S_k = U$ with $|S_i| = a_i + a'_i$. Since the number of such pairs $T_1 \sqcup \cdots \sqcup T_k$ and $T'_1 \sqcup \cdots \sqcup T'_k$ is $\binom{a_1 + \cdots + a_k}{a_1, \dots, a_k} \cdot \binom{a'_1 + \cdots + a'_k}{a'_1, \dots, a'_k}$ and the number of $S_1 \sqcup \cdots \sqcup S_k = U$ with $|S_i| = a_i + a'_i$ is only $\binom{a_1 + \cdots + a_k + a'_1 + \cdots + a'_k}{a_1 + a'_1, \dots, a_k + a'_k}$, we have
\begin{align*}
\binom{a_1 + \cdots + a_k + a'_1 + \cdots + a'_k}{a_1 + a'_1, \dots, a_k + a'_k} \geq \binom{a_1 + \cdots + a_k}{a_1, \dots, a_k} \cdot \binom{a'_1 + \cdots + a'_k}{a'_1, \dots, a'_k}
\end{align*}
as desired.
\end{proof}

\begin{proof}[Proof of Fact~\ref{fact:multichoose-ineq}]
Assume w.l.o.g. that $a_1 \leq a_2 \leq \cdots \leq a_k$. We have
\begin{align*}
\binom{a_1 + \cdots + a_k}{a_1, \dots, a_k} = \prod_{i=1}^k \binom{a_i + \cdots + a_k}{a_i}
&\geq \prod_{i=1}^{\lfloor k/2 \rfloor} \binom{a_i + \cdots + a_k}{a_i} \\
&\geq \prod_{i=1}^{\lfloor k/2 \rfloor} (a_i + \cdots + a_k) \\
&\geq \left(\frac{a_1 + \cdots + a_k}{2}\right)^{\lfloor k/2 \rfloor},
\end{align*}
where the last inequality uses the fact that $a_1 \leq \cdots \leq a_k$.
\end{proof}

\section{Proof of Corollary~\ref{cor:up_bd_DP}}\label{sec:cor_pf}

Corollary~\ref{cor:up_bd_DP} follows from our main theorem (Theorem~\ref{th:up_bd_sec}) and the connection between secure summation protocols and differentially private summation protocols due to Balle et al. \cite{BBGN19}. We recall the latter below.

\begin{lemma}[Lemma 4.1 of \cite{BBGN19}]\label{le:red_sec_to_priv}
Given a $\sigma$-secure protocol in the anonymized setting for $n$-party summation over the domain $\F_q$, where each party sends $f(q,n,\sigma)$ messages each of $g(q, n, \sigma)$ bits, there exists an $(\epsilon, (1+e^{\epsilon}) 2^{-\sigma-1})$-differentially private protocol in the shuffled model for real summation with absolute error $O(1+1/\epsilon)$ where each party sends $f(O(n^{3/2}),n,\sigma)$ messages each of $g(O(n^{3/2}),n,\sigma)$ bits.
\end{lemma}

Corollary~\ref{cor:up_bd_DP} now follows immediately by applying Lemma~\ref{le:red_sec_to_priv} and Theorem~\ref{th:up_bd_sec} with $\sigma = 1 + \log\left(\frac{1 + e^{\varepsilon}}{\delta}\right) = O\left(1 + \varepsilon + \log(1/\delta)\right)$.

We remark here that Lemma~\ref{le:red_sec_to_priv} as stated above is slightly different from Lemma 4.1 of~\cite{BBGN19}. In particular, in~\cite{BBGN19}, the statement requires the secure summation protocol to works for any $\mathbb{Z}_q$ even when $q$ is not a prime power. On the other hand, our analysis in this paper (which uses rank of matrices) only applies to when $q$ is a prime power (i.e., $\F_q$ is a field). However, it turns out that this does not affect the connection too much: instead of picking $q = 2\lceil n^{3/2} \rceil$ as in~\cite{BBGN19}, we may pick $q$ to be the smallest prime larger than $2n^{3/2}$. In this case, $q$ remains $O(n^{3/2})$ and the remaining argument of~\cite{BBGN19} remains exactly the same.

\end{document}